\documentclass{article}

\usepackage{arxiv}

\usepackage[utf8]{inputenc} 
\usepackage[T1]{fontenc}    
\usepackage{url}            
\usepackage[hidelinks]{hyperref}
\usepackage{booktabs}       
\usepackage{amsfonts}       
\usepackage{nicefrac}       
\usepackage{microtype}      
\usepackage{graphicx}
\usepackage{natbib}
\usepackage{doi}
\usepackage{amsmath,amsthm,amssymb}
\usepackage{bm}
\usepackage{xcolor}
\usepackage{algpseudocode}
\usepackage{algorithm}

\graphicspath{{\subfix{../images/}}}

\DeclareMathOperator*{\argmin}{argmin}

\newtheorem{definition}{Definition}[subsection]

\newtheorem{example}{Example}[subsection]

\theoremstyle{plain}
\newtheorem{theorem}{Theorem}[subsection]
\newtheorem{lemma}{Lemma}[subsection]

\newcommand{\R}{\mathbb{R}}

\newcommand{\norm}[1]{\left\lVert#1\right\rVert}
\newcommand{\abs}[1]{\left|#1\right|}
\DeclareMathOperator\md{MD}

\newcommand{\xhat}[1]{\hat{\bm{x}}^{#1}}
\newcommand{\xhatj}[2]{\hat{x}_{#2}^{#1}}

\title{Multivariate outlier explanations using Shapley values and Mahalanobis distances}

\author{\href{https://orcid.org/0000-0002-3430-8308}{\includegraphics[scale=0.06]{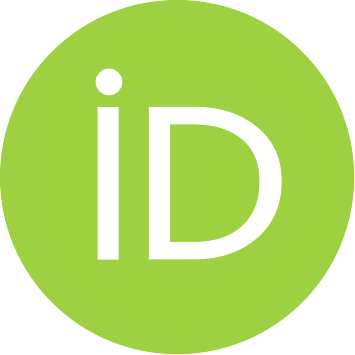}\hspace{1mm}Marcus Mayrhofer}\\
	TU Wien\\
	\texttt{marcus.mayrhofer@tuwien.ac.at} \\
	\And
	\href{https://orcid.org/0000-0002-8014-4682}{\includegraphics[scale=0.06]{orcid.pdf}\hspace{1mm}Peter Filzmoser} \\
	TU Wien\\
	\texttt{peter.filzmoser@tuwien.ac.at} \\
}

\renewcommand{\shorttitle}{Multivariate outlier explanations using Shapley values and Mahalanobis distances}

\hypersetup{
pdftitle={A template for the arxiv style},
pdfsubject={q-bio.NC, q-bio.QM},
pdfauthor={David S.~Hippocampus, Elias D.~Striatum},
pdfkeywords={First keyword, Second keyword, More},
}

\begin{document}
\maketitle

\begin{abstract}
	For the purpose of explaining multivariate outlyingness, it is shown that the squared Mahalanobis distance of an observation can be decomposed into outlyingness contributions originating from single variables. The decomposition is obtained using the Shapley value, a well-known concept from game theory that became popular in the context of Explainable AI. 
    In addition to outlier explanation, this concept also relates to the recent formulation of cellwise outlyingness, where Shapley values can be employed to obtain variable contributions for outlying observations with respect to their ``expected'' position given the multivariate data structure. 
    In combination with squared Mahalanobis distances, Shapley values can be calculated at a low numerical cost, making them even more attractive for outlier interpretation. 
    Simulations and real-world data examples demonstrate the usefulness of these concepts.
\end{abstract}


\section{Introduction}
\label{section:introduction}

Multivariate outlier detection is a topic of unabated popularity in statistics
and computer science. Not only does there exist a wide variety of approaches but also the terminology varies;
anomaly detection, novelty detection, or fraud detection all refer to the problem of identifying unusual behavior \citep{Zimek2018}. In a dataset with $n$ observations measured at $p$ variables, one is interested in identifying observations that do not conform to their expected behavior according to the remaining (neighboring) observations \citep{Chandola2009, Grubbs1969}.

A widespread tool for the detection of multivariate outliers in statistics 
is based on the Mahalanobis distance \citep{mahalanobis1936}. 
Generally, for an observation vector 
$\bm{x} = (x_1,\ldots ,x_p)'$ from a population with expectation vector
$\bm{\mu}=(\mu_1,\ldots ,\mu_p)'$ and 
covariance matrix $\bm{\Sigma}$, the squared Mahalanobis distance of $\bm{x}$ 
to $\bm{\mu}$ with respect to $\bm{\Sigma}$
is given as
\begin{equation}
\label{eq:MD2}
\md_{\bm{\mu},\bm{\Sigma}}^2(\bm{x}) = (\bm{x}-\bm{\mu})'\bm{\Sigma}^{-1}(\bm{x}-\bm{\mu}) ,
\end{equation}
which will be denoted in the following by $\md^2(\bm{x})$.
To specify the outlyingness of an observation from a given sample, the parameters $\bm{\mu}$ and $\bm{\Sigma}$ need to be estimated, with their estimators being denoted as $\hat{\bm{\mu}}$ and $\hat{\bm{\Sigma}}$.
If the underlying distribution is a multivariate normal distribution, it is 
common to use the 0.975 quantile of a chi-square distribution with $p$ degrees of freedom $\chi_{p;0.975}^2$ as a cutoff value \citep{Rousseeuw1990}. Observations with a squared Mahalanobis
distance exceeding this cutoff are identified as multivariate outliers.
It is evident that for outlier detection, the estimates $\hat{\bm{\mu}}$ and $\hat{\bm{\Sigma}}$ themselves must be robust against such outliers. 
Many different proposals for robust estimation of multivariate location and covariance can be found throughout the literature, with one of the most popular being the minimum covariance determinant (MCD) estimator~\citep{Rousseeuw1985}.

The squared Mahalanobis distance from Equation~\eqref{eq:MD2} can also be 
written as
\begin{equation}
\label{eq:MD2element}
\md^2(\bm{x}) =
\sum_{j=1}^p\sum_{k=1}^p (x_j-\mu_j)(x_k-\mu_k)\omega_{jk} ,
\end{equation}
where $\omega_{jk}$ denotes the element $(j,k)$ of the precision matrix $\bm{\Omega} = \bm{\Sigma}^{-1}$.
This outlyingness measure collects distance contributions of
all pairwise variable combinations, weighted by $\omega_{jk}$, resulting in a single number. However, this value cannot be interpreted in the sense of contributions from individual variables, which would be vital for determining the effect of the single variables on the overall outlyingness. 

Analyzing the contributions of individual variables is also of major interest in
Explainable Artificial Intelligence, which is often referred to as Interpretable Machine Learning. For example, suppose a ``black-box'' classifier has been trained on a dataset; it is often essential to know how and why the individual variables of an observation contribute to the model's decision to assign an observation to a particular class \citep{Ribeiro2016}. Various tools have been established for this purpose, and
Shapley values are among the most popular ones. 
Although the Shapley value~\citep{shapley1953} was originally proposed in the context of game theory in \citeyear{shapley1953}, it was applied much later in the context of machine learning by \cite{strumbelj2010, strumbelj2014} and its popularity increased greatly after the publications of \cite{lundberg2017, lundberg2019, lundberg2020}. For a more exhaustive discussion of these methods, we refer to \cite{molnar2019} and \cite{brzemyslaw_2021}. 

In this paper, we propose using the Shapley value for multivariate outlier explanation, which will be directly based on the squared Mahalanobis distance. Our method allows us to determine the individual variable contributions to the outlyingness and to answer the question \textit{why} an observation is flagged as a multivariate outlier.
The arguably most critical disadvantage of the Shapley values in a general setting is their high computational complexity, which exponentially increases with the number of variables. 
However, we will show that the Shapley values resulting from our approach can be expressed as a simplified problem, which substantially facilitates their computation, even in a higher dimension. 
In addition, we present an extension of this concept that enables the assignment of outlyingness scores to pairs of variables, allowing the evaluation of interaction effects. 

It should be mentioned that an alternative approach to answer which variables contribute the most to the multivariate outlyingness of an observation has been presented by \cite{debruyne2019}, who estimate the univariate direction of maximum outlyingness using sparse regression. Nevertheless, this method does not result in an additive decomposition of the squared Mahalanobis distance.

Another approach closely related to outlier explanation is called
cellwise outlier detection. For an overview of this relatively recent research field, we refer to \cite{Raymaekers2021}. Its main idea is to investigate the outlyingness
of each cell of a data matrix instead of focusing on entire observations. 
In general terms, cellwise outlyingness is based on the difference of the actual value of a cell compared to the value we would have expected.

Computing the amount by which a cell is anomalous is also related to multivariate outlier explanation, although the approach is somehow reversed:
The explanations are the result of evaluating the single coordinates of observations and not the result of interpreting the outlyingness of observations in terms of the individual coordinate contributions, which take the covariance structure into account.

The remainder of this paper is structured as follows: In Section~\ref{section:method} we introduce Shapley values before we derive in detail how to apply them for multivariate outlier explanation using squared Mahalanobis distances. Moreover, we outline how to combine those results with the concept of cellwise outlier detection, leading to the cellwise robust outlier explanation algorithms described in Section~\ref{section:algorithm}. The performance of those outlier explanation tools for cellwise outlier detection is demonstrated via the numerical experiments presented in Section~\ref{section:simulations}. In Section \ref{section:applications} we analyze the performance of our method on real-world examples. The final Section~\ref{section:conclusion} summarizes the key points of our findings.

\section{Shapley values for outlier explanation}
\label{section:method}

In the following, we propose a method for the interpretation of multivariate outliers that combines squared Mahalanobis distances with Shapley values \citep{shapley1953}. The concept of Shapley values is briefly introduced based on its nascent field of research, namely cooperative game theory \citep{peters2008}. 

\subsection{Shapley values and cooperative game theory}
\label{subsection:shapley_values_and_cooperative_game_theory}

In cooperative game theory, players can form coalitions that produce a payoff and decide on how their coalitions' proceeds are distributed among them. 

\begin{definition}
A coalitional (cooperative) game with transferable utility (TU-game) $(T,v)$ is given by a set of players $T = \{1,2,\ldots,t\}$ and the characteristic function $v$, which assigns the worth $v(S) \in \R$ to each coalition $S \subseteq T$, such that $v(\emptyset) = 0$.
\label{def:coalitional_game}
\end{definition}

In other words, the function $v$ tells us how much collective payoff a coalition $S$ of players can gain by cooperating. A payoff distribution for the grand coalition $T$ is given by $\bm{\varphi}(v) = (\varphi_1(v),\ldots,\varphi_t(v))'$, where $\varphi_j(v) \in \R$ is the payoff to player $j$. 
There are several proposals on how the payoff should be assigned to the players $j\in T$ to obtain a \textit{fair} distribution. While there are different concepts and notions of fairness in the literature, we will focus on the one introduced by \cite{shapley1953}. The \textit{Shapley value} $\bm{\phi}(v)$, with coordinates
\begin{align}
    \phi_j(v) = \sum_{S \subseteq T \setminus \{j\}} \frac{\abs{S}!(t-\abs{S}-1)!}{t!}\left(v(S\cup\{j\})-v(S)\right),
    \label{eq:shapley_value}
\end{align}
is the \textit{unique} payoff distribution that fulfills 
the following conditions~\citep{young1985}:
\begin{itemize}
\item{\textit{Efficiency:}}
The payoff to individual players $\varphi_j(v)$ must add up to the worth of the grand coalition $v(T)$, hence $\sum_{j = 1}^p \varphi_j(v) = v(T)$.
\item{\textit{Symmetry:}}
If $v(S \cup \{j\}) = v(S \cup \{k\})$ holds for all $S \subseteq T\setminus \{j,k\}$ for two players $j$ and $k$, then $\varphi_j(v) = \varphi_k(v)$.
\item{\textit{Monotonicity:}}
If for any two games $(T,v_1)$ and $(T,v_2)$ and all $S \subseteq T$ the condition 
\begin{align*}
v_1(S \cup \{j\}) - v_1(S) \geq v_2(S \cup \{j\}) - v_2(S)
\end{align*}
is satisfied, then $\phi_j(v_1) \geq \phi_j(v_2)$.
\end{itemize}
Therefore, the Shapley value permits the definition of a fair payoff distribution for the grand coalition $T$.
The term $v(S\cup\{j\})-v(S)$ describes the marginal contribution
of player $j$ to a coalition $S$. The corresponding Shapley value
$\phi_j(v)$ is then given as the weighted
mean of the marginal contributions formed over all possible
different coalitions. 

\subsection{Linking Shapley value and Mahalanobis distance}

Let us consider an observation vector $\bm{x} = (x_1,\ldots ,x_p)'$ from a population with expectation vector $\bm{\mu} = (\mu_1,\ldots ,\mu_p)'$ and covariance matrix $\bm{\Sigma}$.
We would like to investigate the contribution of the $j$-th
coordinate $x_j$ to the outlyingness of $\bm{x}$. 
The set of players is denoted as $P = \{1,\ldots ,p\}$,
and it contains the indices of all variables.
A coalition $S$ is formed by a subset of $P$.
We define the characteristic function $v$ mentioned above as the squared
Mahalanobis distance
\begin{equation}
    \md_{\bm{\mu},\bm{\Sigma}}^2(\xhat{S}) = \md^2(\xhat{S}) \label{eq:characteristic_function}
\end{equation} 
with $\xhat{S}=(\xhatj{S}{1},\ldots ,\xhatj{S}{p})'$ and
\begin{equation}
\label{eq:definexj}
\xhatj{S}{j}:= \begin{cases}
        x_j & \text{if } j \in S\\
        \mu_j & \text{if } j \notin S
    \end{cases},
\end{equation}
which fulfills $\md^2(\xhat{S}) = 0$, if $S=\emptyset$ is the empty set.

In this setting, the $k$-th coordinate of the Shapley value from 
Equation~\eqref{eq:shapley_value} is given as the weighted average of the marginal contributions 
\begin{equation*}
    \Delta_k \md^2(\xhat{S}) := \md^2(\xhat{S\cup\{k\}})-\md^2(\xhat{S})
\end{equation*} 
over all $2^{p-1}$ subsets $S \subseteq P \setminus \{k\}$. This suggests an exponential computational complexity, which becomes costly, especially if $p$ is large. However, the following theorem shows that this highly demanding problem can be reduced to linear complexity.
\begin{theorem} \label{theorem:explain_mahalanobis}
    Given two vectors $\bm{x}, \bm{\mu} \in \R^{p}$ and a non-singular matrix $\bm{\Sigma}  \in \R^{p \times p}$, the contribution of the $k$-th 
    variable to the squared Mahalanobis distance $\md^2(\bm{x})$ based on the Shapley value is given by 
    \begin{align}
             \phi_k(\bm{x},\bm{\mu},\bm{\Sigma}) &:= \sum_{S \subseteq P\setminus\{k\}} \frac{\abs{S}!(p-\abs{S}-1)!}{p!}\Delta_k \md^2(\xhat{S}) \label{eq:shapley_md_long}\\
             &= (x_k-\mu_k) \sum_{j =1}^p (x_j-\mu_j) \omega_{jk}, \label{eq:shapley_md}
    \end{align}
    with $\bm{\Sigma}^{-1} =: \bm{\Omega} = (\omega_{jk})_{j,k = 1,\ldots ,p}$ and $\hat{\bm{x}}^S$ as in Equation~\eqref{eq:definexj}. 
\end{theorem}
\begin{proof}
    The proof of this theorem is given in \ref{section:proof_explain_mahalanobis}.
\end{proof}

Indeed, we can compute the expression of Equation~(\ref{eq:shapley_md}) as an intermediate result when we compute the squared Mahalanobis distance, see Equation~\eqref{eq:MD2element}.

The Shapley value of an observation $\bm{x}$ resulting from Theorem~\ref{theorem:explain_mahalanobis} is given by the vector
\begin{equation}
    \bm{\phi}(\bm{x},\bm{\mu},\bm{\Sigma})=(\phi_1(\bm{x},\bm{\mu},\bm{\Sigma}),\ldots ,\phi_p(\bm{x},\bm{\mu},\bm{\Sigma}))'\label{eq:MDV}
\end{equation}
and we will simply denote it as $\bm{\phi}(\bm{x}) = (\phi_1(\bm{x}),...,\phi_p(\bm{x}))'$, whenever the (robustly estimated) mean and covariance matrix are employed for its computation.
Considering Theorem~\ref{theorem:explain_mahalanobis}, it is straightforward to see that 
\begin{align}
    \bm{\phi}(\bm{x}) &=  (\bm{x}-\bm{\mu}) \circ \bm{\Sigma}^{-1} (\bm{x}-\bm{\mu}) , \label{eq:ShapleyMatrix}
\end{align}
where $\circ$ denotes the element-wise product. 

Since $\bm{\phi}(\bm{x})$ is based on the Shapley value, it is the only decomposition of the squared Mahalanobis distance with the characteristic function defined in  Equation~\eqref{eq:characteristic_function} that fulfills
the following properties:
\begin{itemize}
\item{\textit{Efficiency:}}
The contributions $\phi_j(\bm{x})$, for $j = 1,\ldots ,p$, sum up to the squared Mahalanobis distance of $\bm{x}$, hence 
\begin{equation}
\label{eq:efficiency}
\sum_{j = 1}^p \phi_j(\bm{x}) = \md^2(\bm{x}).
\end{equation}
\item{\textit{Symmetry:}}
If $\md^2(\xhat{S \cup \{j\}}) = \md^2(\xhat{S \cup \{k\}})$ holds for all subsets $S \subseteq P\setminus \{j,k\}$
for two coordinates $j$ and $k$, then $\phi_j(\bm{x}) = \phi_k(\bm{x})$.
\item{\textit{Monotonicity:}}
Let $\bm{\mu}, \tilde{\bm{\mu}} \in \R^p$ be two vectors and $\bm{\Sigma}, \tilde{\bm{\Sigma}} \in \R^{p \times p}$ be two non-singular matrices. If 
\begin{align*}
        \md_{\bm{\mu},\bm{\Sigma}}^2(\xhat{S \cup \{j\}}) - \md_{\bm{\mu},\bm{\Sigma}}^2(\xhat{S}) \geq \md_{\tilde{\bm{\mu}},\tilde{\bm{\Sigma}}}^2(\xhat{S \cup \{j\}}) - \md_{\tilde{\bm{\mu}},\tilde{\bm{\Sigma}}}^2(\xhat{S})
\end{align*}
holds for all subsets $S \subseteq P$, then $\phi_j(\bm{x},\bm{\mu},\bm{\Sigma}) \geq \phi_j(\bm{x},\tilde{\bm{\mu}},\tilde{\bm{\Sigma}})$.
\end{itemize}

A single coordinate $\phi_j(\bm{x})$ of the Shapley value defined in Theorem \ref{theorem:explain_mahalanobis} can be interpreted as the average marginal contribution of the $j$-th variable to the squared Mahalanobis distance of an individual observation $\bm{x}$. 
Thus, $\bm{\phi}(\bm{x})$ decomposes $\md^2(\bm{x})$ into contributions originating from the
single variables. The simplified form presented in the theorem reveals that the larger the value $\phi_j$, the stronger the contribution of the $j$-th coordinate to $\md^2(\bm{x})$. 
It should be noted that the contributions can also be negative, as illustrated in Figure~\ref{fig:2_dim_shapley}.

\begin{figure}[ht!]
    \centering
    \includegraphics[width=0.85\linewidth]{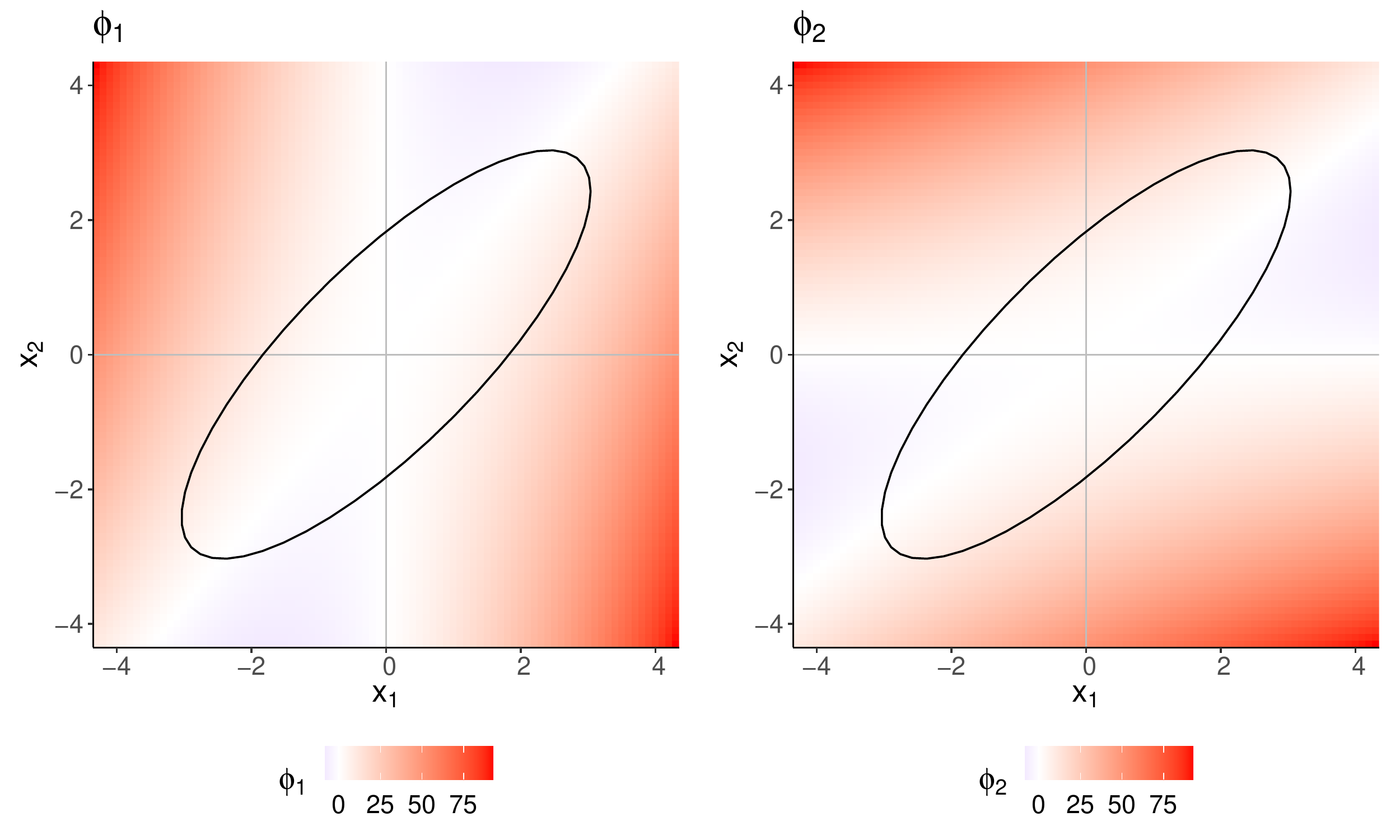}
    \caption{Plots illustrating a two-dimensional visualization of the Shapley values $\bm{\phi}(\bm{x})$ for $\bm{x} \in [-4,4]\times[-4,4]$ with mean $\bm{\mu} = (0,0)'$ and covariance matrix $\bm{\Sigma}$, with elements $\sigma_{12} = \sigma_{21} = 0.8$, and $\sigma_{11} = \sigma_{22} = 1$.
    The graphs are colored according to the components $\phi_1(\bm{x})$ and $\phi_2(\bm{x})$ of the Shapley value, respectively, and both panels show the 99-percentile confidence ellipse.}
    \label{fig:2_dim_shapley}
\end{figure}

\medskip
\noindent
\textbf{Remark:}
The definition given in Equation~\eqref{eq:definexj}, where $\xhatj{S}{j}=\mu_j$ if $j\notin S$, could also be
modified. In the literature, it is often suggested to use the conditional expectation of $x_j$, given all other variables with index contained in $S$, instead of the expected value \citep{lundberg2017}.
However, evaluating the conditional expectation explicitly requires us to impose distributional assumptions or to apply approximation techniques, and this may also lead to high
computational complexity, as for every coordinate $j  \in P$ there are
$2^{p-1}$ possible subsets $S$. Our definition of $\xhat{S}$ results in two major advantages:
\begin{enumerate}
    \item The computational complexity of computing the Shapley value reduces
    from an exponential to a linear one; see Theorem~\ref{theorem:explain_mahalanobis}.
    \item 
    For any $S$ and the resulting $\xhat{S}$, the definition of Equation~\eqref{eq:definexj} results in the fact that
    $\md_{\bm{\mu},\bm{\Sigma}}^2(\xhat{S})$
    is identical to
    $\md_{{\bm{\mu}_S},{\bm{\Sigma}}_S}^2(\bm{x}_S)$,
    where $\bm{x}_S = (x_j)_{j \in S}$ and 
    ${\bm{\mu}}_S = (\mu_j)_{j \in S}$ only consist of the coordinates of $\bm{x}$
    and $\bm{\mu}$ contained in the set $S$, respectively, and 
    ${\bm{\Sigma}}_S$ is the submatrix of $\bm{\Sigma}$ with
    rows and columns included in $S$.
    Therefore, analyzing the outlyingness of $\xhat{S}$ using the squared Mahalanobis distance is equivalent to an analysis of the outlyingness of the 
    lower dimensional version $\bm{x}_S$, see
    also Equation~\eqref{eq:MD2element}.
\end{enumerate}

\subsection{Shapley interaction values}

The Shapley value can also be generalized via the Shapley interaction index~\citep{Grabisch1999,Fujimoto2006}, which is also used in the field of Explainable AI~\citep{lundberg2019}. As the name suggests, the idea is not only to investigate individual variable contributions but also to obtain a measure of interaction between the variables.

Using the notation of cooperative game theory as in 
Section~\ref{subsection:shapley_values_and_cooperative_game_theory}, the Shapley interaction index for $S$, with fixed $\abs{S} = s$, is given by
\begin{align}
    I_{Sh}(v,S)
    = \sum_{T \subseteq P \setminus S} \frac{t!(p-t-s)!}{(p-s+1)!} \Delta_S v(T)\label{eq:shapley_interaction},
\end{align}
with $t = \abs{T}$, and $\Delta_S v(T) = \sum_{L\subseteq S}(-1)^{s-l} v(T\cup L), l = \abs{L}$, also known as the \textit{discrete} or \textit{set function derivative} \citep{Grabisch2016}. We refer to the previously mentioned articles of \cite{Grabisch1999,Fujimoto2006} for more details regarding the theory and properties connected to this concept. 

As before, we decompose the squared Mahalanobis distance using the characteristic function defined in Equation~\eqref{eq:characteristic_function}.
Moreover, we only focus on the \textit{pairwise} Shapley interaction index $(\abs{S} = 2)$, because higher order Shapley interaction indices $(\abs{S} \geq 3)$ turn out to be zero in this setting (see \ref{subsection:proof_explain_mahalanobis_interaction} for a proof).

\begin{theorem} \label{theorem:explain_mahalanobis_interaction}
    Given two vectors $\bm{x}, \bm{\mu} \in \R^{p}$ and a non-singular matrix $\bm{\Sigma}  \in \R^{p \times p}$, the pairwise contributions of the variable pair $(j,k)$ of an observation $\bm{x}$ to the squared Mahalanobis distance $\md^2(\bm{x})$, based on the Shapley interaction index as defined in Equation~\eqref{eq:shapley_interaction}, are collected in the matrix $\bm{\Phi}(\bm{x}) = \bm{\Phi}(\bm{x}, \bm{\mu}, \bm{\Sigma})$, where the off-diagonal elements are given by
    \begin{align}
        \Phi_{jk}(\bm{x}) &:= \sum_{T \subseteq P \setminus \{j,k\}} \frac{t!(p-t-2)!}{(p-1)!} \Delta_{\{j,k\}} \md^2(\xhat{T}) \label{eq:shapley_interaction_long}\\
        &= 2(x_j - \mu_j)(x_k-\mu_k)\omega_{jk} \label{eq:shapley_interaction_simplified},
    \end{align}
    with $ \Delta_{\{j,k\}} \md^2(\xhat{T}) = \md^2(\xhat{T \cup \{j,k\}}) - \md^2(\xhat{T\cup \{j\}}) - \md^2(\xhat{T\cup \{k\}}) + \md^2(\xhat{T})$. The diagonal elements are defined as 
    \begin{align}
        \Phi_{jj}(\bm{x}) &:= \phi_j(\bm{x}) - \sum_{k \neq j} \Phi_{jk}(\bm{x}) \\
        &= (x_j-\mu_j)^2 \omega_{jj} - (x_j-\mu_j)\sum_{k \neq j} (x_k - \mu_k) \omega_{jk}, \label{eq:shapley_interaction_long_diagonal}
    \end{align} 
    where $\phi_j(\bm{x})$ is the $j$-th coordinate of the Shapley value as in Theorem \ref{theorem:explain_mahalanobis}.
\end{theorem}
\begin{proof}
    The proof of this theorem is given in~\ref{section:proof_explain_mahalanobis_interaction}.
\end{proof}

This theorem reveals that the computation of the Shapley interaction index for pairwise contributions is straightforward. 
It also shows that $\Phi_{jk}(\bm{x})$ measures the marginal deviation of the $j$-th and the $k$-th coordinate from their mean, weighted by the corresponding entry of the precision matrix.
The sign of $\Phi_{jk}(\bm{x})$ not only depends on the sign of $\omega_{jk}$, but also on 
the signs of $(x_j - \mu_j)$ and $(x_k-\mu_k)$, as illustrated in Figure~\ref{fig:2_dim_visuals_interaction_md}. 

\begin{figure}[ht!]
    \centering
    \includegraphics[width=0.85\linewidth]{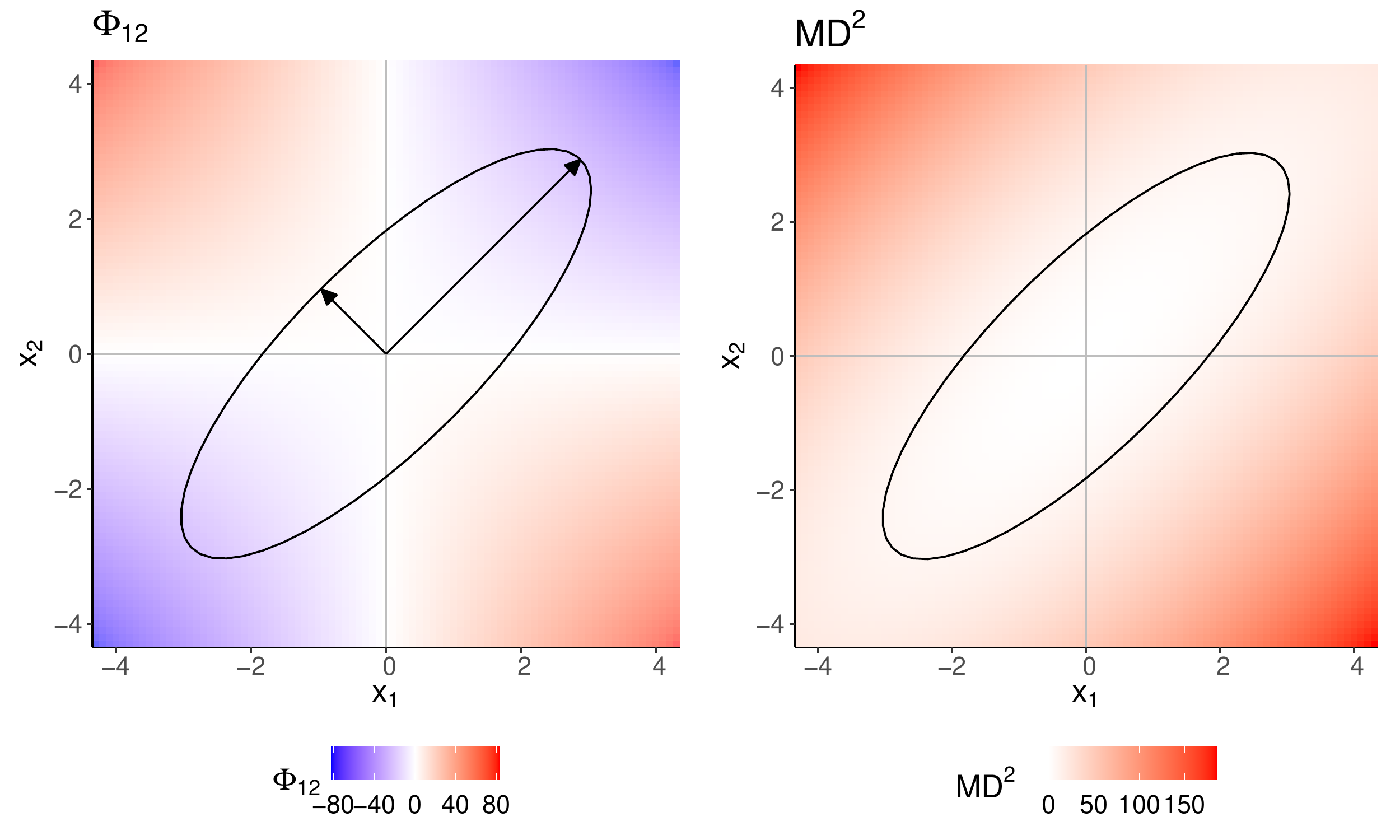}
    \caption{Using the same setup as for the example described in Figure~\ref{fig:2_dim_shapley}, the pairwise contributions $\Phi_{12}$ of $x_1$ and $x_2$ to the squared Mahalanobis distance are visualized in the left panel. In this simple two-dimensional example, we can see that the pairwise contributions are highest in the direction of the eigenvector of $\bm{\Sigma}$ with the smallest eigenvalue. Hence, observations with a large multivariate outlyingness and a small univariate outlyingness are assigned high pairwise outlyingness scores $\Phi_{12}$. In the right graph, we display the squared Mahalanobis distance, and both panels include the 99-percentile confidence ellipse.}
    \label{fig:2_dim_visuals_interaction_md}
\end{figure}

The definition of the diagonal elements $\Phi_{jj}(\bm{x})$ is chosen such that a generalization of the \textit{Efficiency} property given in Equation~\eqref{eq:efficiency} is possible: 
\begin{align*}
    \phi_j(\bm{x}) = \sum_{k=1}^p \Phi_{jk}(\bm{x}) \quad \text{and} \quad \md^2(\bm{x}) = \sum_{j=1}^p\sum_{k=1}^p \Phi_{jk}(\bm{x}).
\end{align*}
Thus, the Shapley values for every variable can be decomposed into pairwise
interactions with the remaining variables.

It is worth mentioning that there are other suggestions on how to generalize the Shapley value such that an explicit definition of $\Phi_{jj}(\bm{x}), j = 1, \ldots, p$, is not necessary \citep[e.g.][\emph{Shapley-Taylor interaction index}]{sundararajan2020shapley}.

\section{Cellwise robust outlier explanation}
\label{section:algorithm}

Cellwise outlier detection focuses on identifying unusual \textit{cells} rather than rows in a data matrix. Such a procedure is particularly justified when dealing with datasets containing many variables: If only individual cells of an observation are contaminated, then the majority of non-contaminated cells still contains valuable information that should not be discarded. Moreover, already a small proportion of outlying cells spread out over the whole data matrix could, in a rowwise treatment, soon lead to a setting where the majority of observations are considered as traditional rowwise outliers. 
The cellwise contamination model has first been formalized by \citet{alqallaf2009propagation}. Several papers that build on this concept are discussed  in \citet{Raymaekers2021}, they also introduce a novel procedure for cellwise outlier identification.

As already outlined in Section~\ref{section:introduction}, the key objective of this work concerns the explanation of multivariate outliers based on the Shapley value for given (or appropriately estimated\footnote{The proportion of contaminated rows may quickly exceed $50\%$ in the cellwise contamination setting \citep{alqallaf2009propagation}. However, rowwise robust methods can only deal with settings where at least half of the observations are not corrupted. To obtain initial cellwise robust covariance estimates, the 2SGS approach of \cite{Agostinelli2015}, the DDC method of \cite{Rousseeuw2018}, or the cellMCD estimator of \cite{Raymaekers2022} can be used.}) 
parameters $\bm{\mu}$ and $\bm{\Sigma}$. Since it enables an additive decomposition of the squared Mahalanobis distance, the Shapley value can be used to identify outlying cells.
However, these contributions do not inform us about the supposed cell values under the assumption that they were not contaminated.
Apart from detecting outlying cells, estimating the values the cells were supposed to have is of major importance when handling cellwise outliers.
In this section, we outline how to combine the ideas of cellwise outlier detection and multivariate outlier explanation to obtain \textit{cellwise robust outlier explanations.}

\subsection{SCD (Shapley Cell Detector) algorithm}

As a starting point, we take another look at the decomposition derived in Theorem~\ref{theorem:explain_mahalanobis}, where we obtain the average marginal contributions of each component to the squared Mahalanobis distance. Equations~\eqref{eq:definexj} and \eqref{eq:shapley_md_long} allow us to interpret said contributions in more detail:
The value of $\phi_j(\bm{x})$ represents the average change in $\md^2(\bm{x})$ across all $2^{p-1}$ possible variations of other variables, when the $j$-th component of $\bm{x}$ is replaced by its mean. Hence, positive values of $\phi_j(\bm{x})$ indicate that replacing $x_j$ with $\mu_j$ would lead to an average reduction in $\md^2(\bm{x})$, whereas negative values indicate that such a replacement would have the opposite effect. 

The information provided by the Shapley value can now be used to design an algorithm for identifying outlying cells and replacing their values. We propose a stepwise procedure, which is described in detail in
Algorithm~\ref{algorithm:SCD}. We call this method \textit{Shapley Cell Detector}, abbreviated as SCD. The set $S$ is updated in each step and will finally contain the indices of all cells of an observation $\bm{x}$ which are marked as outlying. 
In the course of each iteration,
we replace the coordinates of $\bm{x}$ that have the highest scores
according to the Shapley value $\bm{\phi}(\bm{x})$, until the modified observation $\tilde{\bm{x}}$ is no longer a multivariate outlier. 
The replaced value does not directly correspond to the mean, but rather to a value towards the direction of the mean whereby the magnitude of the correction is controlled by a step size parameter $\delta \in (0,1]$.
This is done for each set $S$ until a score resulting from the complement $\bar{S}:= P \setminus S$ of $S$ is larger than one obtained from the set $S$.
Here, the $s$-dimensional subvector $\tilde{\bm{x}}_S = (\tilde{x}_j)_{j \in S}$ of the modified observation $\tilde{\bm{x}}$ consists of the replaced values, which are dependent on $\bm{\mu}_S = (\mu_j)_{j \in S}$. 
Note that the maximum in line 6 of Algorithm~\ref{algorithm:SCD} is usually unique, implying that $k=1$ and only one index is added to $S$ per iteration.

\begin{algorithm}
    \caption{Shapley Cell Detector (SCD)}
    \label{algorithm:SCD}
    \begin{algorithmic}[1] 
        \Procedure{SCD}{$\bm{x}, \bm{\mu}, \bm{\Sigma}, \delta$} 
            \State $\tilde{\bm{x}} \gets \bm{x}$
            \State $S \gets \emptyset$
            \State $\bm{\phi}=(\phi_1,\ldots ,\phi_p)' \gets \bm{\phi}(\bm{x}, \bm{\mu}, \bm{\Sigma})=(\phi_1(\bm{x}, \bm{\mu}, \bm{\Sigma}),\ldots ,\phi_p(\bm{x}, \bm{\mu}, \bm{\Sigma}))'$
            \While{$\md^2(\tilde{\bm{x}}) > \chi^2_{p,0.99}$} 
                \State $S \gets S \cup \{j_1,\ldots ,j_k\} \text{, where }(\phi_{j_l})_{l=1,\ldots ,k}=\max_{i=1,\ldots ,p}\phi_i$
                \While{$\max_{j \in S}\phi_j > \max_{j \in \bar{S}}\phi_j $} 
                    \State $\tilde{\bm{x}}_S \gets \tilde{\bm{x}}_S - (\tilde{\bm{x}}_S - \bm{\mu}_S)\delta$
                    \State $\bm{\phi} \gets \bm{\phi}(\tilde{\bm{x}}, \bm{\mu}, \bm{\Sigma})$
                \EndWhile
            \EndWhile \label{euclidendwhile}
            \State \textbf{return} $\tilde{\bm{x}}$
        \EndProcedure
    \end{algorithmic}
\end{algorithm}

\begin{example} \label{example:5d}
We illustrate the working principle of Algorithm~\ref{algorithm:SCD} by considering a
$5$-dimensional observation $\bm{x} = (0,1,2,2.2,2.5)'$ from a population with mean $\bm{\mu} = (0,0,0,0,0)'$ and covariance matrix $\bm{\Sigma}$, with elements $\sigma_{jk} = 0.9, j \neq k$, and $\sigma_{jj}=1$. Here, $\bm{x}$ would be marked as a multivariate outlier since $\md^2(\bm{x}) = 44.90 > 15.09 =\chi^2_{5,0.99}$ and we can employ the Shapely value of Theorem~\ref{theorem:explain_mahalanobis} to explain this multivariate outlier, resulting in $\bm{\phi}(\bm{x}) = (0, -5.07, 9.87, 15.26, 24.84)'$. 
Those outlyingness scores are then used in Algorithm~\ref{algorithm:SCD} to flag outlying cells and, for simplicity, we analyze the case where $\delta=1$. In this scenario, the coordinate $x_5$ is identified first, followed by $x_4$ and then $x_3$. Each variable in turn is replaced by $\mu_5, \mu_4$, and $\mu_3$, respectively.
This results in an altered version $\tilde{\bm{x}}$ of the original observation $\bm{x}$, which is no longer outlying, and therefore the algorithm stops.

It should be noted that in this example, we have no information about which cells are truly outlying or have been manipulated. However, in general, it seems desirable to keep the number of modified coordinates as small as possible.
\end{example}

Algorithm~\ref{algorithm:SCD} is easy to implement and fast to compute. The discrepancy between the original and replaced cells indicates the amount of outlyingness in the particular variables.
However, this simplicity results from our definition of the Shapley value in Theorem~\ref{theorem:explain_mahalanobis}, which leads to a replacement by a value towards the mean in Algorithm~\ref{algorithm:SCD}.

Figure~\ref{fig:my_label} provides a further illustration of the SCD procedure for a two-dimensional example. It schematically displays five specific observations, denoted by \texttt{A} to \texttt{E}, to which Algorithm~\ref{algorithm:SCD} is applied. The left plot shows the result when setting $\delta=1$ in the algorithm, while the right plot corresponds to $\delta=0.1$. The points in the plots highlight the individual computation steps of the algorithm, and the ellipse indicates the stopping criterion $\chi^2_{2,0.99}$. While for $\delta=1$ the algorithm uses at most two steps, this behavior changes for the case of $\delta=0.1$. Using a smaller step size leads to different replacement values for the points \texttt{B}, \texttt{D}, and \texttt{E}. Comparing the computation steps for points \texttt{B} and \texttt{D}, the results in the right plot seem more meaningful since they avoid increasing the Mahalanobis distance during the computation, and the final replacement is more similar to the original points.

\begin{figure}
    \centering
    \includegraphics[width = 0.85\linewidth]{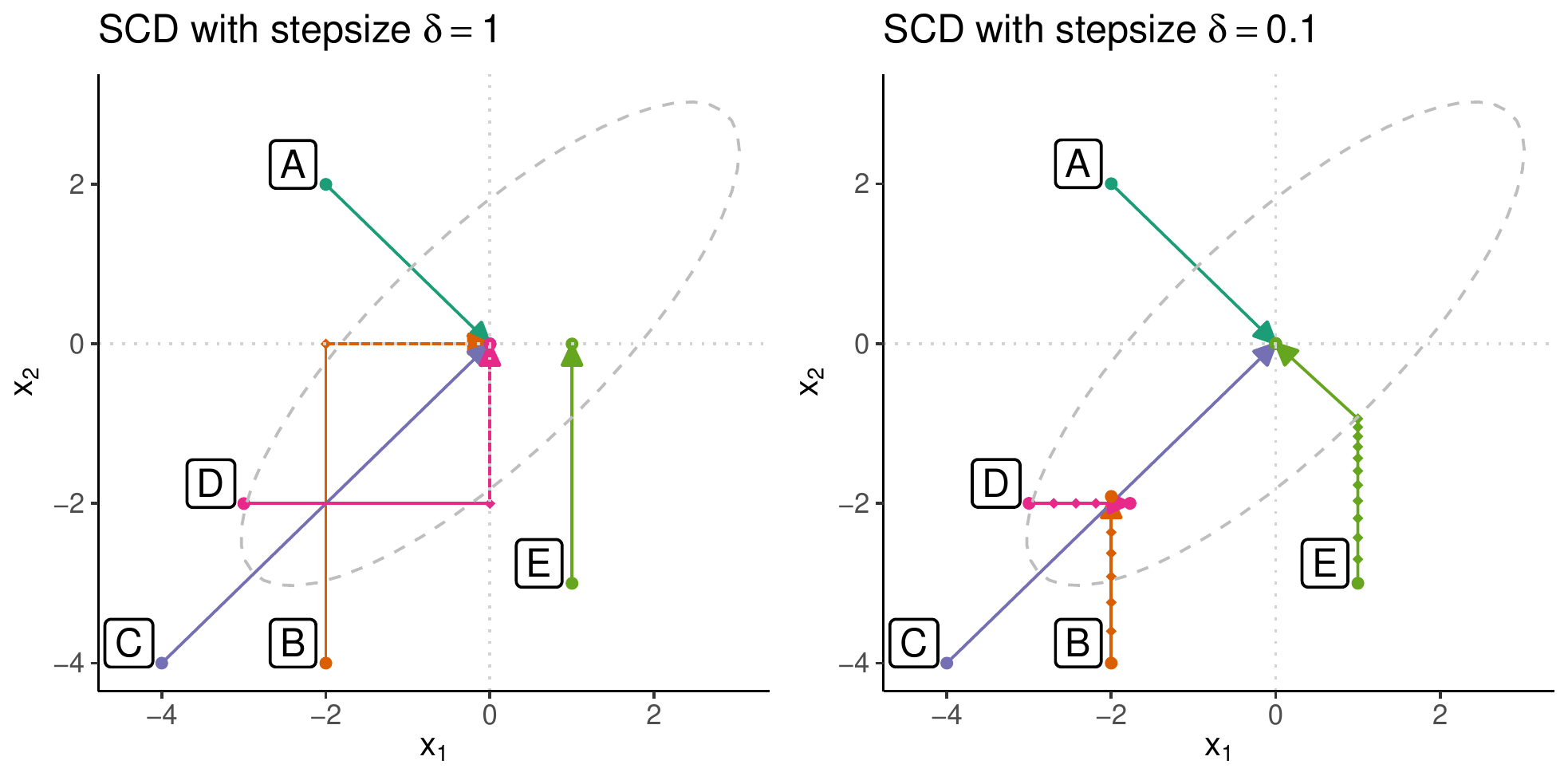}
    \caption{In this figure, two graphs are displayed to illustrate the operating principle of Algorithm~\ref{algorithm:SCD} in a two-dimensional setting. Both plots show the position of the five outlying points \texttt{A} to \texttt{E} and their replacements. The plot on the left side shows the results when cells are directly replaced by their corresponding mean, while the right side illustrates the stepwise approach.}
    \label{fig:my_label}
\end{figure}

Until now, we only considered a replacement of outlying cells by the mean or by a value towards the direction of the mean. However, the algorithm is only stopped by a sufficient reduction of the squared Mahalanobis distance. 
Therefore, the task at hand can thus be redefined further: Find the optimal replacements for outlying cells to achieve the highest possible reduction in squared Mahalanobis distance. As before, the Shapley value should determine the outlyingness of the cells. 

\subsection{MOE (Multivariate Outlier Explainer) algorithm}
\label{subsection:replacement}

Based on the definition of the Shapley value in Equation~\eqref{eq:shapley_md}, a coordinate has a low outlyingness contribution if it is close to its mean. Consequently, it is unlikely that this cell is flagged as outlying. This center-outward ordering is induced by the squared Mahalanobis
distance computed with respect to the mean, and thus it explains the \textit{global outlyingness} of an observation. However, the described procedure might not be optimal for detecting cellwise outliers,
where \textit{local outlyingness} is emphasized, because the information contained in the regular cells of an observation could be incorporated to define an optimal replacement. 
For this purpose, an alternative approach to using the mean as the center for computing Mahalanobis distances and Shapley values is outlined in the following paragraphs. We call the newly proposed center parameter \textit{reference point}. This new Shapley value will also be used later for an outlier replacement strategy.

The question of how to best replace cells of an observation to minimize the squared Mahalanobis distance has been addressed in \cite{Raymaekers2021}. Here we assume that the set $S$ of outlying cells is fixed (and $S\neq \emptyset$) for an observation $\bm{x}=(x_1,\ldots , x_p)'$, and the cells $x_j$ should be shifted to the values $\tilde{x}_j$, for $j \in S$. Explicitly we can write this as $\bm{x} - \bm{E}_S \bm{\beta}$ where $\bm{E}_S$ denotes the $p \times s$ matrix with the standard basis vectors $\bm{e}_j, j \in S$ as columns.
The squared Mahalanobis distance of this expression can now be rewritten as follows,
\begin{align*}
        \md^2_{\bm{\mu},\bm{\Sigma}}(\bm{x} - \bm{E}_S \bm{\beta}) 
    &= (\bm{x} - \bm{\mu}- \bm{E}_S \bm{\beta})'\bm{\Sigma}^{-1}(\bm{x} - \bm{\mu}- \bm{E}_S \bm{\beta}) \\
    &= \norm{\bm{\Sigma}^{-1/2}(\bm{x} - \bm{\mu}- \bm{E}_S \bm{\beta})}_2^2\\
    &= \norm{\bm{\Sigma}^{-1/2}(\bm{x} - \bm{\mu})- \bm{\Sigma}^{-1/2}\bm{E}_S\bm{\beta} }_2^2,
\end{align*}
where $\bar{S}:=P \setminus S$.
Minimizing this expression corresponds to a least-squares problem, which leads to the least-squares
estimator
\begin{align}
    \hat{\bm{\beta}}(S)=\argmin_{\bm{\beta} \in \R^{s}} \md^2_{\bm{\mu},\bm{\Sigma}}(\bm{x} - \bm{E}_S \bm{\beta}) = (\bm{E}_S' \bm{\Sigma}^{-1} \bm{E}_S)^{-1} \bm{E}_S' \bm{\Sigma}^{-1}(\bm{x} - \bm{\mu}), \label{eq:minMD}
\end{align}
and the replaced values are given by $\tilde{\bm{x}}_S=\bm{x}_S-\hat{\bm{\beta}}(S)$ \citep{Raymaekers2021}.
If $S$ consists of only one element, say $S=\{j\}$, for $j\in P$, then the solution of Equation~(\ref{eq:minMD}) simplifies to 
\begin{align}
    \hat{\beta}(j)=\frac{1}{\omega_{jj}} (\omega_{j1},\dots ,\omega_{jp})(\bm{x} - \bm{\mu}), \label{eq:minMD1}
\end{align}
where $\omega_{ij}$ denotes the element $(i,j)$ of $\bm{\Sigma}^{-1}$, and
the modification for observation $\bm{x}$ is given by $\tilde{x}_j=x_j-
\hat{\beta}(j)$.

Building on those findings, we can now define the new reference point $\tilde{\bm{\mu}}(\bm{x},S)$ for a fixed set of outlying cells $S$, by setting each coordinate to
\begin{align}
    \tilde{\mu}_j(\bm{x},S) = x_j-\hat{\beta}_{(j)}(S \cup \{j\}), \label{eq:reference_point}
\end{align}
where $\hat{\beta}_{(j)}(S \cup \{j\})$ is the component of $\hat{\bm{\beta}}(S \cup \{j\})$ corresponding to the index $j$. To determine the set $S$, we adapt the SCD procedure, by incorporating $\tilde{\bm{\mu}}(\bm{x},S)$ as a reference point for the Mahalanobis distance and updating it in each iteration. We refer to this procedure as Multivariate Outlier Explainer (MOE) and outline its general workflow in Algorithm~\ref{algorithm:MOE}.

\begin{algorithm}
    \caption{Multivariate Outlier Explainer (MOE)}
    \label{algorithm:MOE}
    \begin{algorithmic}[1] 
        \Procedure{MOE}{$\bm{x}, \bm{\mu}, \bm{\Sigma}, \delta$} 
            \State $\tilde{\bm{x}} \gets \bm{x}$
            \State $S \gets \emptyset$
            \State $\bm{d} = (d_1,\ldots,d_p)' \gets (0,\ldots,0)'$
            \State $\tilde{\bm{\mu}} = (\tilde{\mu}_1,\ldots ,\tilde{\mu}_p)'\gets \tilde{\bm{\mu}}(\bm{x}, S) = 
            (x_1-\hat{\beta}(1),\ldots ,x_p-\hat{\beta}(p))'$
            \State $\bm{\phi}=(\phi_1,\ldots ,\phi_p)' \gets \bm{\phi}(\tilde{\bm{x}}, \tilde{\bm{\mu}}, \bm{\Sigma})=(\phi_1(\tilde{\bm{x}}, \tilde{\bm{\mu}}, \bm{\Sigma}),\ldots ,\phi_p(\tilde{\bm{x}}, \tilde{\bm{\mu}}, \bm{\Sigma}))'$
            \While {$\md_{\tilde{\bm{\mu}}, \bm{\Sigma}}^2(\tilde{\bm{x}}) > \chi^2_{p,0.99}(\md^2(\tilde{\bm{\mu}}))$} \label{line:conv_crit} 
                \State $S \gets S \cup \{j_1,\ldots ,j_k\} \text{, where } (\phi_{j_l})_{l=1,\ldots ,k}=\max_{i=1,\ldots ,p}\phi_i$
                \While{$\max_{j \in S}\phi_j > \max_{j \in \bar{S}}\phi_j $} 
                    \State $\bm{c} \gets (\tilde{\bm{x}}_S - \tilde{\bm{\mu}}_S)\delta$
                    \State $\bm{d}_S \gets \bm{d}_S + \bm{c}$ \label{line:distance_update}
                    \State $\tilde{\bm{x}}_S \gets \tilde{\bm{x}}_S - \bm{c}$
                    \State $\bm{\phi} \gets \bm{\phi}(\tilde{\bm{x}}, \tilde{\bm{\mu}}, \bm{\Sigma})$
                \EndWhile
                \State $\tilde{\bm{\mu}} \gets \tilde{\bm{\mu}}(\bm{x}, S)$\label{line:update_reference}
            \EndWhile
            \State $\bm{d} = (d_1,\ldots,d_p)' \gets (d_1/\sqrt{\sigma_{11}},\ldots,d_p/\sqrt{\sigma_{pp}})'$
            \State $S \gets \{j_1,\ldots ,j_m\} \text{, for which } (d_{j_l})_{l = 1,\ldots,m} > \eta \max_{i=1,\ldots ,p}d_i$                
            \State $\tilde{\bm{\mu}} \gets \tilde{\bm{\mu}}(\bm{x}, S)$
            \State $\bm{\phi} \gets \bm{\phi}(\bm{x}, \tilde{\bm{\mu}}, \bm{\Sigma})$\label{line:update_phi}
            \State $\tilde{\bm{x}} \gets \bm{x}$
            \State $\tilde{\bm{x}}_S \gets \tilde{\bm{\mu}}_S$
            \State \textbf{return} $\tilde{\bm{x}}, \tilde{\bm{\mu}}, \bm{\phi}$
        \EndProcedure
    \end{algorithmic}
\end{algorithm}

The MOE procedure is initialized by computing the reference point $\tilde{\bm{\mu}} = \tilde{\bm{\mu}}(\bm{x}, S)$, with $S = \emptyset$. For the initial computation of $\hat{\bm{\beta}}$ we can simply apply Equation~\eqref{eq:minMD1} to each coordinate of $\bm{x}$, which can be done in one step by matrix multiplication.
Using this initial reference point, we obtain the squared Mahalanobis distance $\md_{\tilde{\bm{\mu}},\bm{\Sigma}}^2(\tilde{\bm{x}})$, which is in turn used to define the corresponding Shapley value $\bm{\phi}(\tilde{\bm{x}}, \tilde{\bm{\mu}},\bm{\Sigma})$ according to Equation~(\ref{eq:ShapleyMatrix}).\footnote{The properties of the Shapley value listed in Section~\ref{section:method} remain unchanged, particularly the \textit{Efficiency} property:  The sum of the coordinates of the Shapley value equals the squared Mahalanobis distance with respect to the new reference point.} 
Outlying cells are then identified based on the Shapley value and corrected in the direction of their corresponding entries of $\tilde{\bm{\mu}}$, resulting in the modified observation $\tilde{\bm{x}}$. The process of updating the reference point $\tilde{\bm{\mu}}$, identifying outlying cells based on their Shapley values, and correcting them in the direction of $\tilde{\bm{\mu}}$, is then repeated until the vector $\tilde{\bm{x}}$ is no longer marked as outlying. Aside from using the reference point $\tilde{\bm{\mu}}$ in the MOE procedure instead of $\bm{\mu}$, the concept of the algorithm is similar to the SCD procedure, but there are two other important distinctions:
\begin{itemize}
\item The outlier cutoff value used in line \ref{line:conv_crit} is adapted to the new reference point. \cite{Filzmoser2014} have shown that for a sample $\bm{x}$ drawn from a multivariate normal distribution $\mathcal{N}(\bm{\mu},\bm{\Sigma})$, the conditional distribution of the squared Mahalanobis distance $\md_{\tilde{\bm{\mu}}, \bm{\Sigma}}^2(\bm{x})$ given $\tilde{\bm{\mu}}$ is a non-central chi-square distribution with $p$ degrees of freedom and non-centrality parameter $\lambda = \md^2(\tilde{\bm{\mu}})$, denoted as $\chi^2_p(\lambda)$. Therefore, the 0.99 quantile of this distribution is taken as the cutoff value to exit the loop.
\item Since the goal of this procedure is cellwise outlier detection, we want to avoid flagging coordinates which were only shifted by a negligible amount. Therefore, we monitor the distance $\bm{d}$ by which each cell of $\bm{x}$ is shifted in the direction of $\tilde{\bm{\mu}}$. Initially, this distance is set to $d_j = 0, j = 1,\ldots,p$, followed by an iterative update of the distance variable in line~\ref{line:distance_update}. Moreover, we adjust $\bm{d}$ such that the distances are independent of the scale of the single coordinates. We then update the set of outlying coordinates $S$ by only choosing coordinates for which $d_j > \eta\max_{l=1,\ldots,p} d_l$, with $\eta \in [0,1]$. As a default value, we have selected $\eta = 0.2$ since it represents be a good trade-off between the recall and the precision of the procedure. Finally, we amend $\tilde{\bm{\mu}}(\bm{x}, S)$,  $\bm{\phi}(\bm{x}, \tilde{\bm{\mu}}, \bm{\Sigma})$, and $\tilde{\bm{x}}$ according to the updated set $S$. 
\end{itemize}

Algorithm~\ref{algorithm:MOE} allows us to detect and impute cellwise outliers and it also yields a local explanation of the outlyingness. Furthermore, the Shapley values computed with respect to the reference point $\tilde{\bm{\mu}}(\bm{x},S)$ can be used to explain the results of other cellwise outlier detection procedures. To this end, we merely need to compute $\tilde{\bm{\mu}} = \tilde{\bm{\mu}}(\bm{x},S)$ for a given set of outlying cells $S$ of an observation $\bm{x}$. By subsequently determining the Shapley value $\bm{\phi}(\bm{x}, \tilde{\bm{\mu}}, \bm{\Sigma})$, we can therefore explain \textit{why} the observation is outlying.

\begin{example}
    We reiterate Example~\ref{example:5d} with the MOE procedure, using a step size of $\delta$ of $0.1$. The first two coordinates $x_1$ and $x_2$ are marked as outlying, resulting in $\tilde{\bm{\mu}} = \tilde{\bm{\mu}}(\bm{x},\{1,2\}) = (2.19, 2.19, 2.27, 2.13, 2.04)$ and $\bm{\phi}(\bm{x},\tilde{\bm{\mu}},\bm{\Sigma}) = (34.89, 7.07, -0.86, 1.28, 4.88)$.
    
    Comparing the results of Algorithms~\ref{algorithm:SCD} and \ref{algorithm:MOE}, it can be seen that the sets of outlying cells for the two algorithms are disjoint, and the interpretation of the result is therefore different. The reason for this discrepancy is mainly that we no longer decompose $\md^2_{\bm{\mu}, \bm{\Sigma}}(\bm{x})$, but instead the squared Mahalanobis distance of the amended reference point $\tilde{\bm{\mu}}$, $\md^2_{\tilde{\bm{\mu}}, \bm{\Sigma}}(\bm{x})$. 
    While the Shapley value $\bm{\phi}(\bm{x}, \bm{\mu}, \bm{\Sigma})$ used in Algorithm~\ref{algorithm:SCD} explains the global outlyingness, the Shapley value $\bm{\phi}(\bm{x}, \tilde{\bm{\mu}}, \bm{\Sigma})$ used in Algorithm~\ref{algorithm:MOE} provides us with a local understanding of the outlyingness, which is better suited to the setting of cellwise outlyingness. 
    
    In Figure \ref{fig:introduction_example_new} we compare the final Shapley values yielded by the SCD and MOE algorithms, respectively. In Figure~\ref{fig:introduction_example_phi_history} we show the Shapley values computed during each iteration for both algorithms, using a step size $\delta = 0.1$. Both figures indicate the squared Mahalanobis distance (black bar) and the corresponding (non-)central chi-square quantile (dotted line). 
    \begin{figure}
        \centering
        \includegraphics[width=0.85\linewidth]{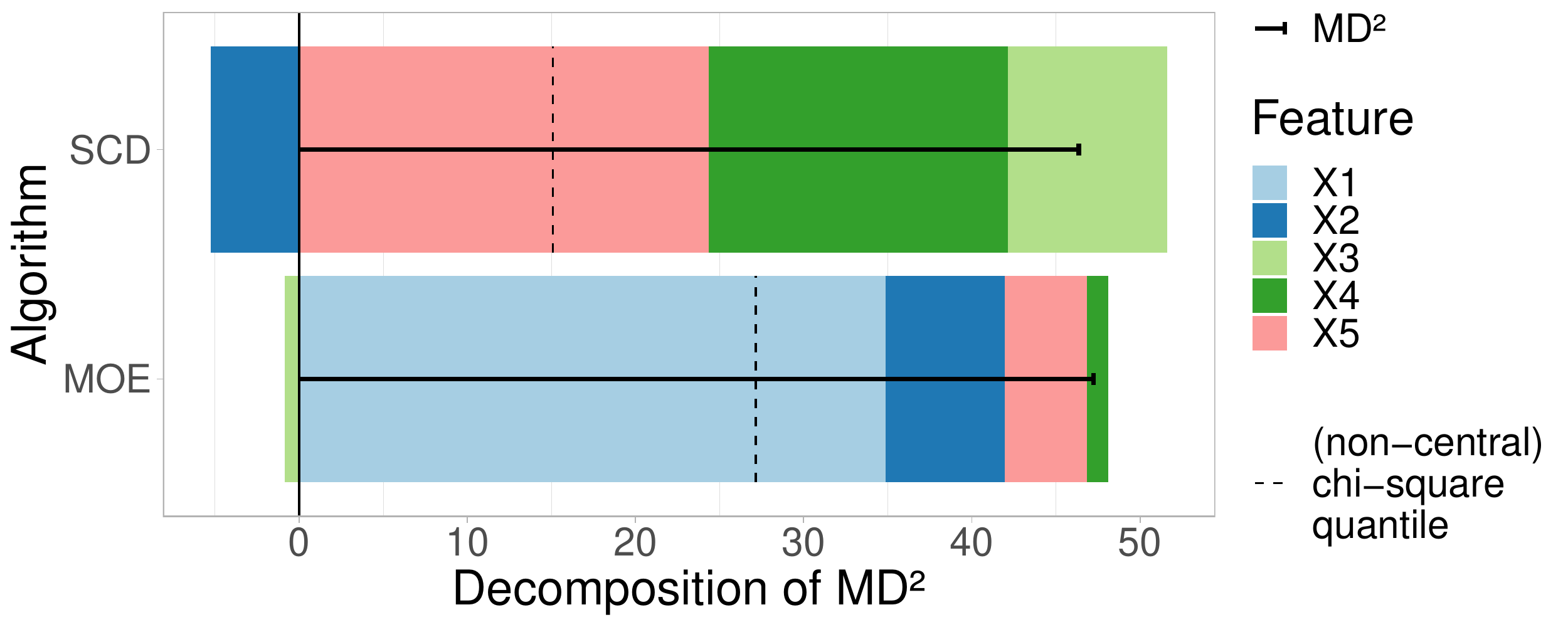}
        \caption{Comparison of the Shapley values $\bm{\phi}(\bm{x}, \bm{\mu}, \bm{\Sigma})$ used in Algorithm~\ref{algorithm:SCD} to explain the \textit{global} outlyingness, and $\bm{\phi}(\bm{x}, \tilde{\bm{\mu}}, \bm{\Sigma})$, used in Algorithm~\ref{algorithm:MOE} to gain \textit{local} insights on the outlyingness, with the input values defined in Example~\ref{example:5d}. The SCD procedure identifies the three coordinates $x_3$, $x_4$, and $x_5$, which are furthest from the mean $\bm{\mu}$. On the other hand, the MOE algorithm uses the alternative reference point $\tilde{\bm{\mu}}$ to identify variables $x_1$ and $x_2$.}
        \label{fig:introduction_example_new}
    \end{figure}
    \begin{figure}
        \centering
        \includegraphics[width=1\linewidth]{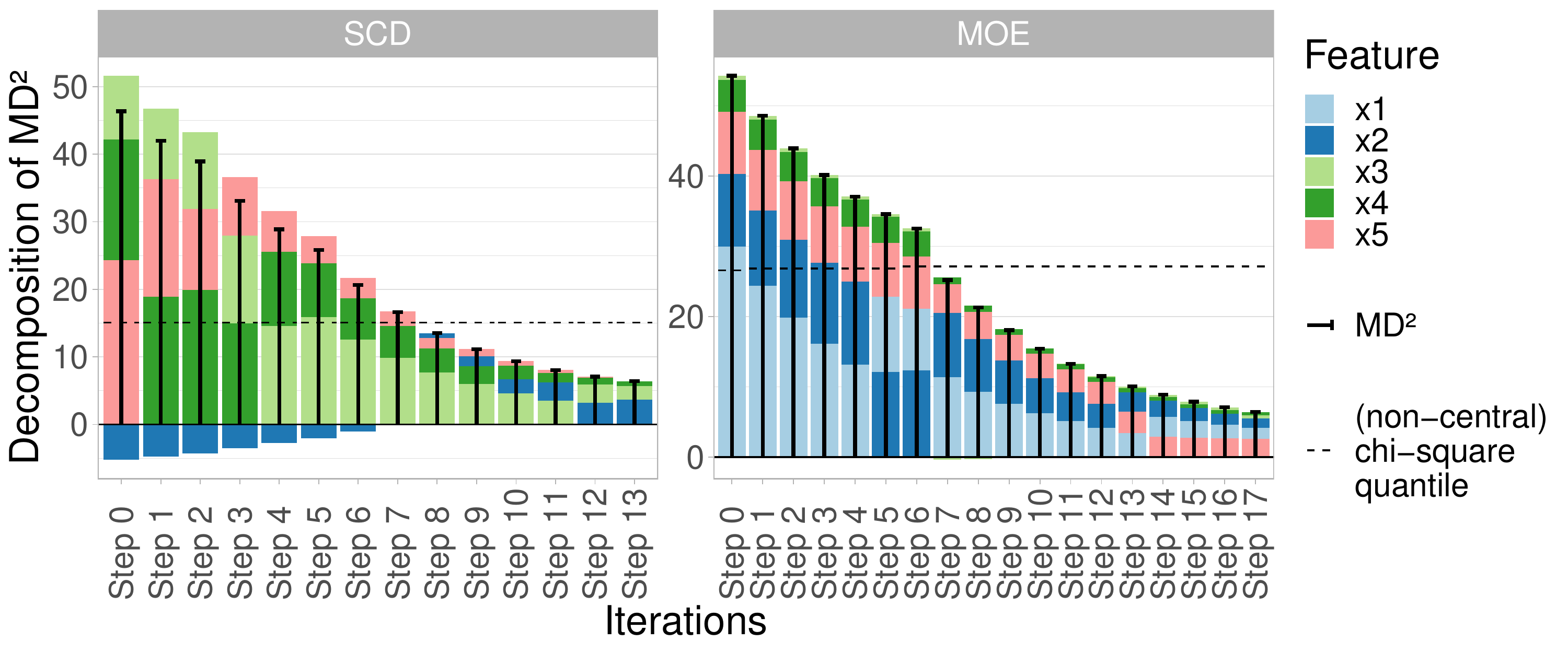}
        \caption{Comparison between the Shapley values calculated for each iteration of Algorithm~\ref{algorithm:SCD} (left) and Algorithm~\ref{algorithm:MOE} (right), respectively, for Example~\ref{example:5d}. While the outlyingness is monotonically decreasing in both cases, the sets of identified variables are disjoint. Both the SCD and MOE procedures reduce the outlyingness by iteratively shifting the identified variables toward the corresponding coordinates of $\bm{\mu}$ or $\tilde{\bm{\mu}}$, respectively.    
        }
        \label{fig:introduction_example_phi_history}
    \end{figure}    
\end{example}

\section{Simulations}
\label{section:simulations}

The simple numerical example from the previous section has illustrated that the SCD and MOE algorithms can lead to quite different outcomes. However, it needs to be emphasized that their purposes also differ: While the SCD procedure aims at global outlier explanation, i.e.~with respect to the distribution of the entire dataset, the MOE procedure is locally applicable and builds on the local information contained in the regular cells of an individual observation. Nevertheless, it can be interesting to compare both procedures in terms of their ability to identify cellwise outliers and, in particular, to examine their performance in comparison to a reference method, namely the cellHandler procedure introduced by \cite{Raymaekers2021}.\footnote{We choose standard parameters for all three procedures, meaning that both the SCD and MOE algorithms are set up with a step size of $\delta = 0.1$ and the MOE procedure additionally uses a detection threshold $\eta = 0.2$.}

In our analysis, we compare two different mechanisms for generating outliers and analyze the effects of various parameter configurations, which are summarized in Table~\ref{tab:simulation_parameters} and described in more detail in the following paragraphs. For each specific parameter combination, we repeat the simulations 50 times and compute averages of the resulting measures Recall, Precision, and F-Score.
\begin{table}[!h]
\caption{Summary of the parameters used for the two simulation scenarios on cellwise outlier detection discussed in Section~\ref{section:simulations}.}
\centering
    \begin{tabular}[t]{lrr}
    \toprule
     Parameters & Shift outliers & Structured outliers\\
     \midrule
     Dimension, $p$ & $5, 10, 20, 30, 40$ & $5, 10, 20, 30, 40$\\
     Covariance, $\bm{\Sigma}$ & $\bm{C}_{\text{mix}}$, $\bm{C}_{\text{low}}$, $\bm{C}_{\text{mod}}$& $\bm{C}_{\text{mix}}$, $\bm{C}_{\text{low}}$, $\bm{C}_{\text{mod}}$\\
     Fraction of outlying columns, $\epsilon_1$ & $0.1, 0.2, 0.3, 0.4$ & -\\
     Fraction of outlying rows, $\epsilon_2$ & $0.1, 0.2, 0.3, 0.4$ & -\\
     Fraction of outlying cells, $\epsilon_3$ & - & $0.1, 0.2, 0.3, 0.4$\\
     Magnitude of outlyingness, $\gamma$ & $1, 2, 3$ & $2, 3, 4, 5, 6$\\
     \midrule
     Total combinations & 720 & 300\\
     \bottomrule
    \end{tabular}
    \label{tab:simulation_parameters}
\end{table}

For both outlier generation procedures, we generate data matrices with $p$ columns and $n= 20p$ rows from multivariate normal distributions with mean $\bm{\mu}=\bm{0}$ and three different types of covariance matrices $\bm{\Sigma}$, namely $\bm{C}_{\text{mod}}$, $\bm{C}_{\text{mix}}$, and $\bm{C}_{\text{low}}$. In all three cases, the diagonal elements are set to 1. 
For $\bm{C}_{\text{mod}}$, the off-diagonal elements are chosen as $0.5$, resulting in moderate correlations. The off-diagonal elements of $\bm{C}_{\text{mix}}$ correspond to $(-0.9)^{\abs{j-k}}, j \neq k$, yielding both high and low correlations. 
For $\bm{C}_{\text{low}}$, the off-diagonal elements of are randomly generated as described in \cite{Agostinelli2015}, generally resulting in low correlations. 

To analyze the effect of highly correlated shift-outliers, we randomly select $\lceil n \epsilon_2 \rceil$ rows, and for each of those rows we replace $r = \lceil p \epsilon_1 \rceil$ randomly selected cells by $r$-variate outliers. Those follow a Gaussian distribution with mean $\bm{\mu} = (\gamma, \ldots, \gamma)'$ and covariance matrix $\tilde{\bm{\Sigma}}$, with elements $\tilde{\sigma}_{jk} = 0.7, j\neq k$, and $\tilde{\sigma}_{jj} = 1$.
The magnitude of the outliers is determined by the value $\gamma$,
which is selected according to Table~\ref{tab:simulation_parameters}. Following this approach, the fraction of outlying cells ranges between $0.01$ and $0.16$.

For the second scenario, outliers are generated such that they are structurally outlying, but have low univariate outlyingness, as proposed by \cite{Raymaekers2021}. For this purpose, $n \epsilon_3$ cells are selected randomly in each column. Like this, each row contains a subset $K \subseteq P$ of cells $\bm{x}_K$ which are subsequently replaced by the vector $\gamma \sqrt{k}\bm{u}'/\md_{\bm{\mu}_K,\bm{\Sigma}_K}(\bm{u})$, where $k=|K|$, and $\bm{u}$ is the eigenvector of $\bm{\Sigma}_K$ that corresponds to the smallest eigenvalue.\footnote{We want to mention that computational issues arose with the cellHandler procedure for the case of moderate to high correlations and $\gamma = 2$. To allow for a fair comparison, we chose to exclude the cases in which these issues occurred from all three procedures.}

We summarize the overall results in Table~\ref{tab:simulation_results}, comparing Precision, Recall, and F-Score. The performance metrics are averaged over all parameters not listed in the table ($p$, $\epsilon_1$, $\epsilon_2$, $\epsilon_3$, and $\gamma$) and all replications.
Regarding Precision, the MOE algorithm exhibits the best results in 5 out of 6 settings. Concerning Recall, the SCD procedure performs best when the correlations are low to moderate, while the cellHandler procedure performs best when the correlations are moderate or mixed. Finally, when comparing the F-Score, we see that each algorithm outperforms the remaining two at least once. However, the results listed in the table are averaged over a wide range of parameter settings, therefore we study the individual effects of the different parameters in more detail in the following. 

\begin{table}[!h]

\caption{Summary of the results of the simulations described in Section~\ref{section:simulations}. The performance metrics Precision, Recall, and F-Score listed in this table are averaged over all replications and parameter combinations.}
\centering
    \small
    \begin{tabular}[t]{llrrrrrr}
        \toprule
        \multicolumn{1}{c}{ } & \multicolumn{1}{c}{ } & \multicolumn{3}{c}{Shift outliers} & \multicolumn{3}{c}{Structured outliers} \\
        \cmidrule(l{3pt}r{3pt}){3-5} \cmidrule(l{3pt}r{3pt}){6-8}
        $\bm{\Sigma}$ & Algorithm & Precision & Recall & F-Score & Precision & Recall & F-Score\\
        \midrule
        $\bm{C}_{mix}$ & SCD & 0.690 & 0.737 & 0.708 & 0.546 & 0.551 & 0.540\\
        $\bm{C}_{mix}$ & MOE & 0.894 & 0.707 & \textbf{0.782} & 0.916 & 0.545 & \textbf{0.668}\\
        $\bm{C}_{mix}$ & cellHandler & 0.760 & 0.743 & 0.741 & 0.854 & 0.564 & 0.667\\
        \addlinespace
        $\bm{C}_{low}$ & SCD & 0.713 & 0.510 & \textbf{0.574} & 0.767 & 0.715 & \textbf{0.729}\\
        $\bm{C}_{low}$ & MOE & 0.678 & 0.396 & 0.478 & 0.880 & 0.597 & 0.695\\
        $\bm{C}_{low}$ & cellHandler & 0.599 & 0.473 & 0.508 & 0.900 & 0.630 & 0.722\\
        \addlinespace
        $\bm{C}_{mod}$ & SCD & 0.767 & 0.405 & 0.507 & 0.859 & 0.530 & 0.627\\
        $\bm{C}_{mod}$ & MOE & 0.808 & 0.421 & \textbf{0.528} & 0.954 & 0.476 & 0.599\\
        $\bm{C}_{mod}$ & cellHandler & 0.649 & 0.471 & 0.522 & 0.917 & 0.513 & \textbf{0.634}\\
        \bottomrule
    \end{tabular}
    \label{tab:simulation_results}
\end{table}

In Figure~\ref{fig:simulations_marginal_dim} we analyze the effect of the dimension $p$ on the cellwise outlier detection performance. We focus on the case of highly correlated shift outliers, with fixed $\epsilon_1 = \epsilon_2 = 0.4$ and $\gamma = 1$. This results in a difficult situation, with many moderately contaminated
cells. For all three algorithms and covariance structures, we observe an increase in Precision as $p$ increases. The SCD procedure shows the strongest increase and the highest overall Precision in case of low correlations. For the mixed and moderate correlations, the MOE procedure exhibits the highest Precision. Moving on to Recall, we see an initial increase for all three methods for mixed and moderate correlations, and a decrease for low correlations. While the Recall is similar for all methods in case of moderate and mixed correlations, we observe that the SCD procedure has the highest Recall in case of low correlations. 

For the structured outliers, we illustrate the influence of $\gamma$ for fixed $\epsilon_3 = 0.4$ and $p = 30$ in Figure~\ref{fig:simulations_cellwise_gamma}. As expected, Precision and Recall are increasing as the magnitude of outlyingness, controlled by $\gamma$, increases. The MOE procedure shows the highest overall Precision. However, regarding Recall, we see that the SCD procedure performs better for mixed and high correlations. For low correlations, the cellHandler procedure exhibits the steepest increase in Recall as $\gamma$ increases. 

In conclusion, these simulations show that our approaches based on the Shapley value, particularly the MOE procedure, yield comparable results to one of the current state-of-the-art methods, namely the cellHandler procedure. 
While cellwise outlier detection presents the focus of the latter method, our approach is instead based on utilizing cellwise outlier detection specifically to enhance and robustify the outlyingness scores based on Theorem~\ref{theorem:explain_mahalanobis}, with respect to an observation's ``expected'' position, as outlined in Equations~\eqref{eq:minMD} and \eqref{eq:reference_point}. 
\begin{figure}
    \centering
    \includegraphics[width=1\linewidth]{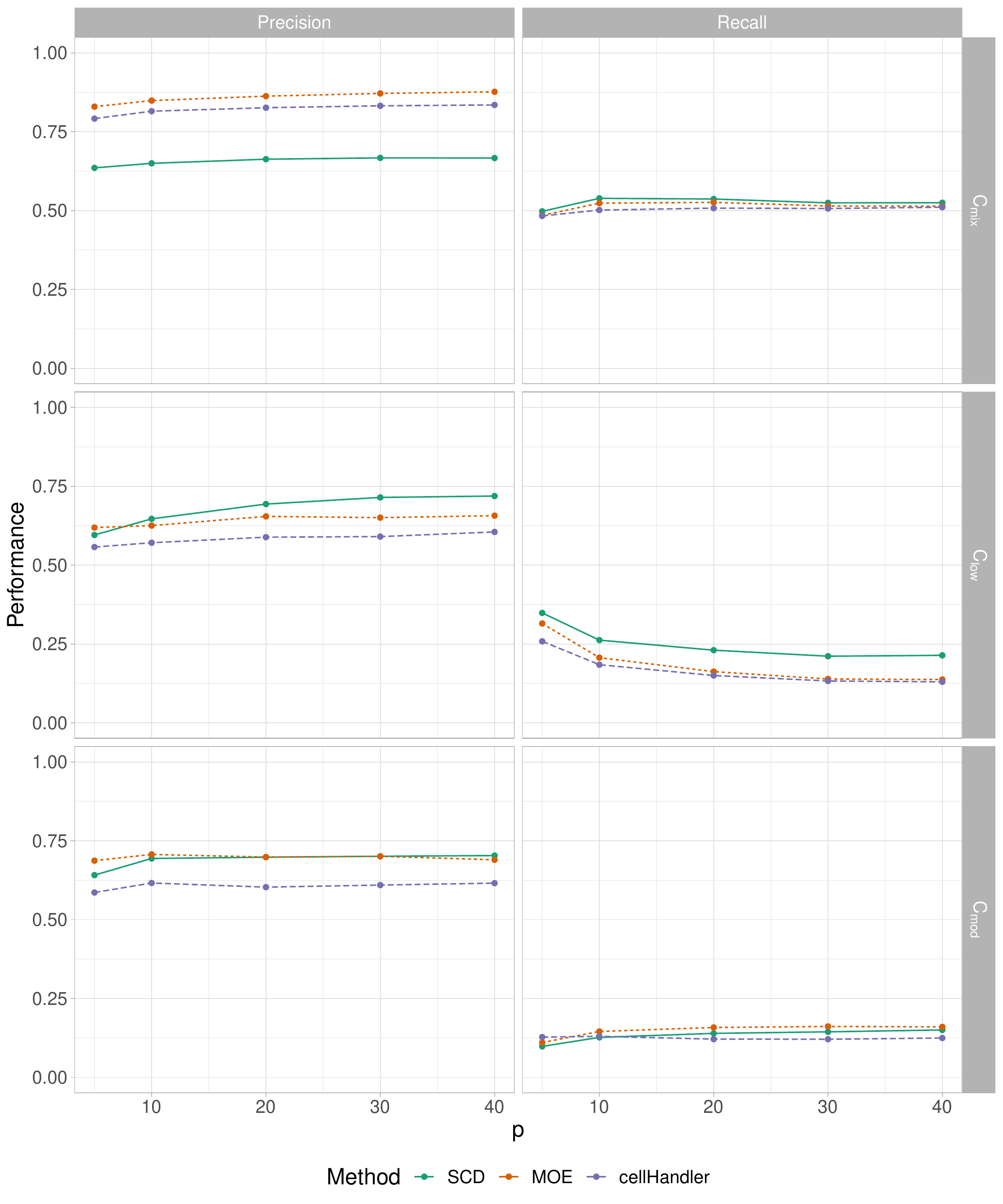}
    \caption{Comparison between the SCD, MOE, cellHandler procedures in the simulation setting of cellwise shift outliers outlined in Section~\ref{section:simulations}, with simulation parameters $\epsilon_1 = \epsilon_2 = 0.4$ and $\gamma = 1$. The performance scores Precision (left) and Recall (right) of the individual algorithms are listed separately for each type of covariance structure.}
    \label{fig:simulations_marginal_dim}
\end{figure}
\begin{figure}
    \centering
    \includegraphics[width=1\linewidth]{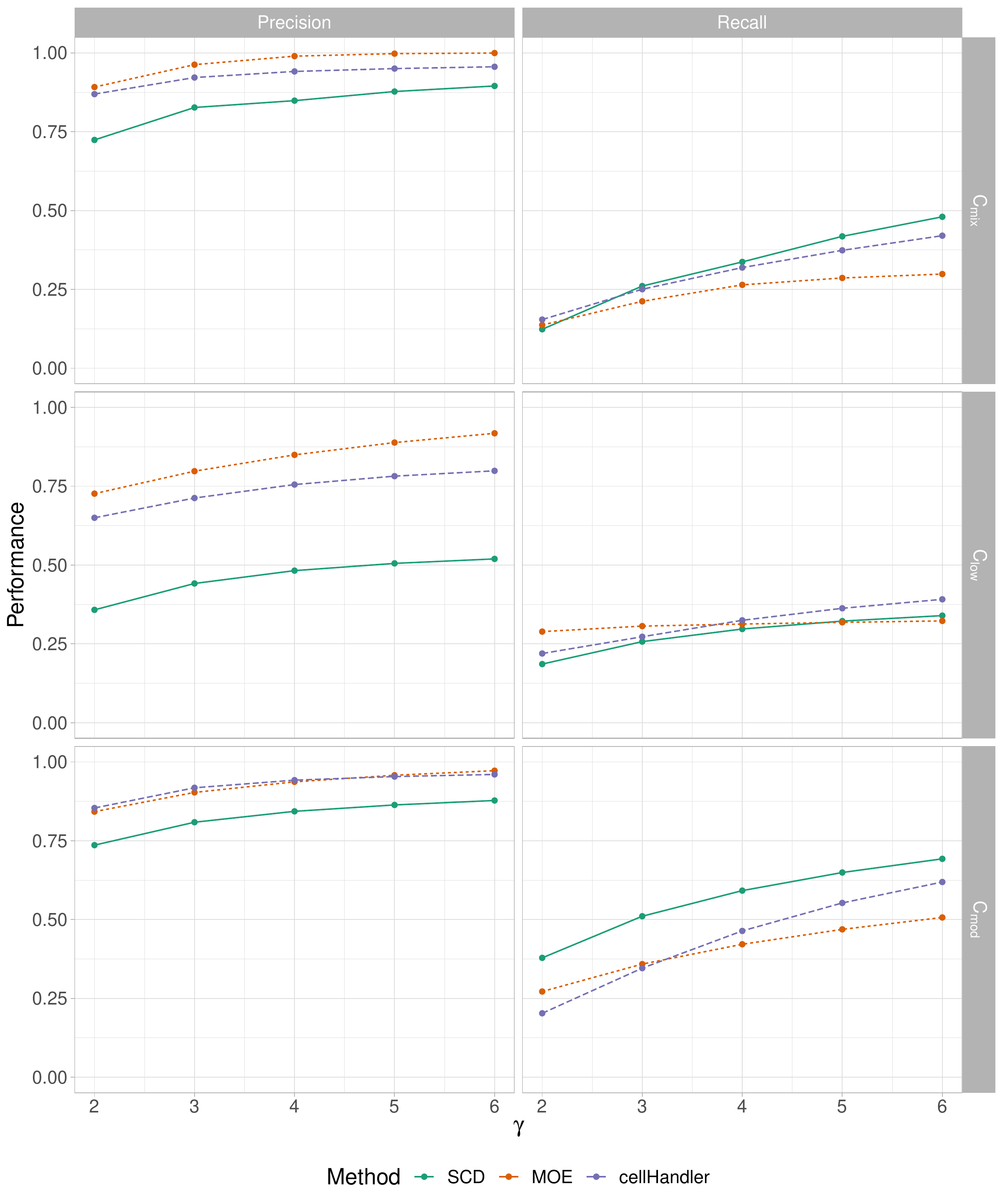}
    \caption{Comparison between the SCD, MOE, cellHandler procedures in the simulation setting of structured cellwise outliers outlined in Section~\ref{section:simulations}, with simulation parameters $\epsilon_3 = 0.4$ and $p = 30$. The performance scores Precision (left) and Recall (right) of the individual algorithms are listed separately for each type of covariance structure.}
    \label{fig:simulations_cellwise_gamma}
\end{figure}

\section{Applications}
\label{section:applications}

While the simulations shown in the previous section have demonstrated the performance of the methods and algorithms introduced in Section~\ref{section:method} and \ref{section:algorithm} on simulated datasets, we now apply them to two real-world data. To this end, we analyze the \textit{Top Gear} dataset from \cite{Alfons2021} and the \textit{Weather in Vienna} dataset from \cite{weather}. 

\subsection{Top Gear}

The \textit{Top Gear} dataset comprises measurements of 11 numerical attributes (see Figure~\ref{fig:TopGear_new_phi}) of 245 complete data instances of cars featured on the website of the BBC television series. Concerning data preprocessing, we apply a logarithmic transformation to five variables to obtain more symmetrical marginal distributions. Additionally, each column is robustly centered and scaled based on the median and the MAD. Furthermore, we estimate the covariance using the MCD estimator before applying the SCD, MOE, and cellHandler procedures.

In the following, we use three different types of plots to analyze the results of all three tested algorithms on this dataset: Figure~\ref{fig:TopGear_new_phi} summarizes the Shapley values, Figure \ref{fig:TopGear_new_cells} shows the outlying cells, and Figure~\ref{fig:TopGear_interaction} displays the Shapley interaction indices, respectively.

\begin{figure}[ht!]
    \centering
    \includegraphics[width=1\linewidth]{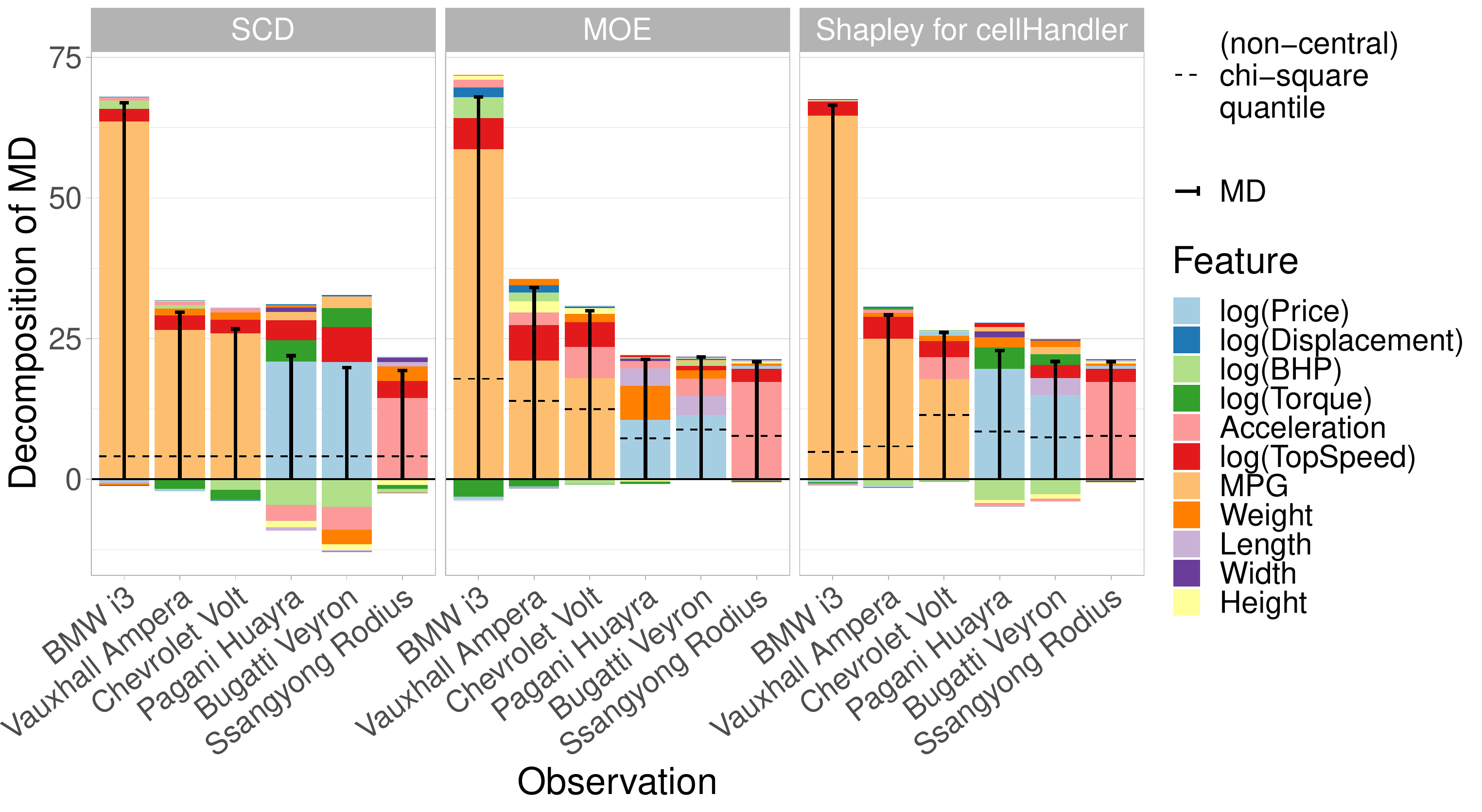}
    \caption{Comparison of the outlyingness scores resulting from the SCD (left), MOE (center), and cellHandler (right) procedures. Each graph shows a visualization of the Shapley values for the six most outlying observations. }
    \label{fig:TopGear_new_phi}
\end{figure}

In detail, Figure \ref{fig:TopGear_new_phi} consists of three graphs, each displaying the outlyingness decompositions according to the applied algorithm, of the six cars with the highest Mahalanobis distance.\footnote{If we are analyzing multiple observations with large differences in squared Mahalanobis distance, plotting the squared distance is ineligible and we display the square root instead. However, we are decomposing the squared distance in Theorem~\ref{theorem:explain_mahalanobis}, therefore we need to scale the outlyingness scores. For this reason, we derive each variable's proportional contribution to the squared distance and multiply it by the (not-squared) Mahalanobis distance. While this results in a somewhat distorted graph, this workflow enables us to analyze and compare multiple observations using a stacked bar chart.} 
In the left panel, we see the results generated using the SCD procedure, where we use the center of the data as a reference point. In the center panel, we show the results of using the MOE algorithm with the non-central chi-square cutoff. For both procedures, we use a step size $\delta = 0.1$ and the detection threshold for the MOE algorithm is defined as $\eta = 0.2$. In the right panel, we show the results of using the cellHandler procedure to flag outlying cells, and then employing the Shapley value, with reference point $\tilde{\bm{\mu}}(\bm{x},S)$ according to Equation~(\ref{eq:minMD}), to enhance the interpretability of the results, as outlined in Section~\ref{section:algorithm}.

Analyzing Figure~\ref{fig:TopGear_new_phi}, we first want to focus on the three cars with the highest outlyingness. For those cars, the main contribution to the squared Mahalanobis distance in the three graphs is caused by the variable \texttt{MPG}. Considering that these specific models are hybrid vehicles, it seems reasonable that their fuel consumption differs strongly from that of gasoline and diesel cars. All three methods lead to similar results in this case.
For the two sports cars \texttt{Bugatti Veyron} and \texttt{Pagani Huayra}, we see that the \texttt{Price} variable is contributing the most to the outlyingness, which is again visible in the results of all three methods. For these two cars, most characteristics are similar to a certain extent, except for their weight: The Bugatti weighs $1990$ kg and the Pagani has only a weight of $1350$ kg. This fact becomes clearly visible when applying the MOE algorithm, where the \texttt{Weight} variable has a high contribution to the squared Mahalanobis distance of the Pagani, but not for the Bugatti. 
The three procedures are again in agreement for the \texttt{Ssangyong Rodius}, where \texttt{Acceleration} contributes the most. In fact, the listed value for \texttt{Acceleration} is $0$, which is clearly an error in the published dataset itself.

\begin{figure}[ht!]
    \centering
    \includegraphics[width=1\linewidth]{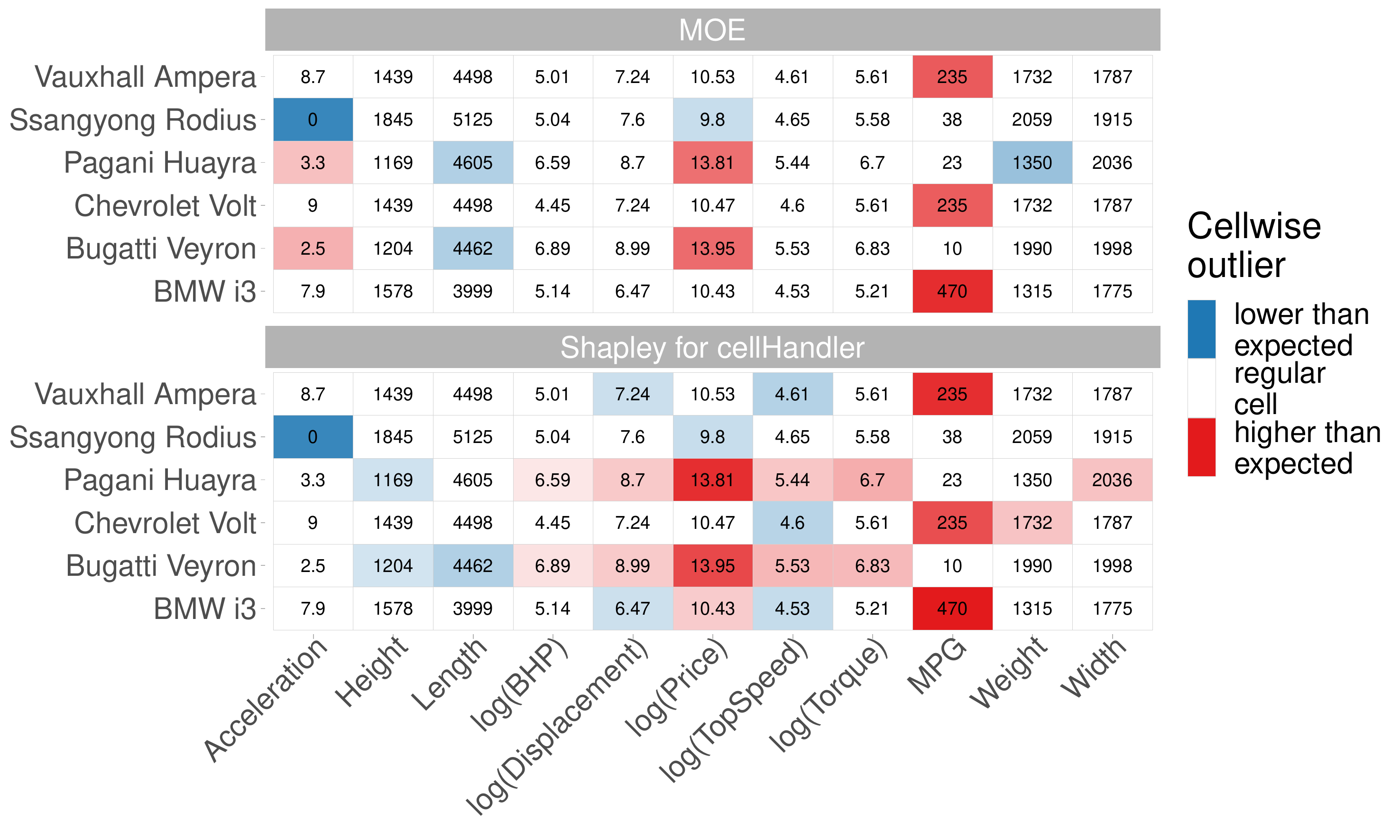}
    \caption{Outlying cells according to Algorithm~\ref{algorithm:MOE} (top) and the cellHandler procedure (bottom). Each cell shows the original value from the dataset, color coding indicates whether those values were higher (red) or lower (blue) than the imputed values, and the color intensity is based on the magnitude of the Shapley value.}
    \label{fig:TopGear_new_cells}
\end{figure}

In Figure~\ref{fig:TopGear_new_cells}, we show the results of applying the MOE procedure (top) to the TopGear dataset, as well as the Shapley values based on the cellHandler procedure (bottom). In these plots, the original values of the variables are displayed in each cell. Regular cells are represented by white rectangles, while outlying cells are colored red or blue, depending on whether the cell's original value is higher (red) or lower (blue) than the replacement. 
The color intensity is given according to the Shapley values of the cells. 
The biggest differences between the MOE and the cellHandler algorithm can be seen between the two sports cars \texttt{Bugatti Veyron} and \texttt{Pagani Huayra}, where the cellHandler procedure results in many more outlying cells. However, it is surprising that the \texttt{Acceleration} parameter is not flagged, since both cars have an exceptionally fast acceleration.

Finally, Figure~\ref{fig:TopGear_interaction} consists of heatmaps displaying the Shapley interaction indices, and barplots showing the corresponding Shapley values for the \texttt{Chevrolet Volt} (left) and \texttt{Pagani Huayra} (right). The Shapley values and interaction indices are based on the reference point we obtain from Algorithm~\ref{algorithm:MOE}. For the Chevrolet, we see a single outstanding index for \texttt{MPG}. On the other hand, the Pagani not only shows a high index for \texttt{Price}, but also for the pairwise outlyingness score between \texttt{Weight} and \texttt{Price}, which indicates that for an expensive sports car it is unexpectedly lightweight. 

\begin{figure}[ht!]
    \centering
    \includegraphics[width=0.85\linewidth]{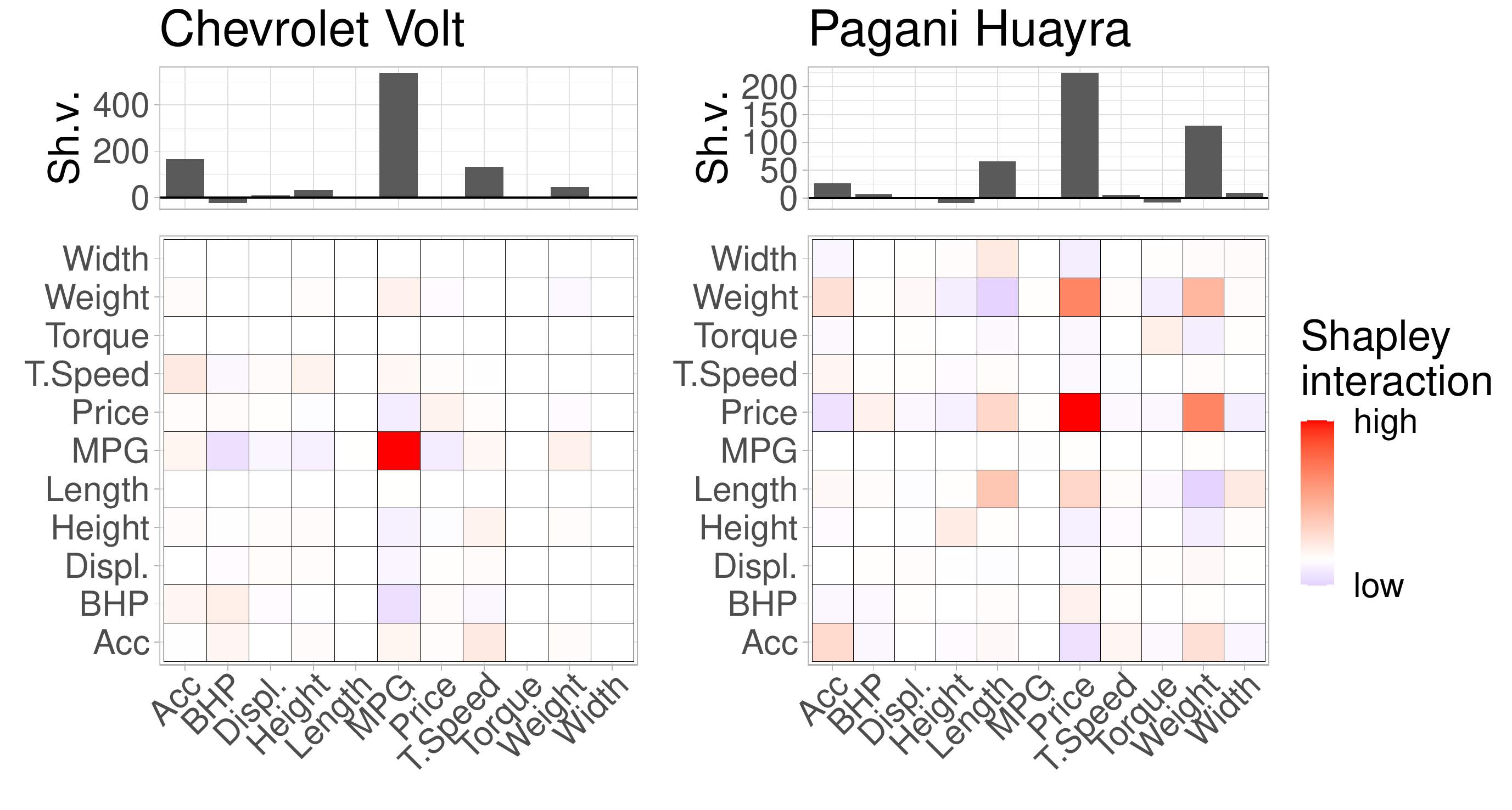}
    \caption{The two graphs in the lower portion of this figure show the Shapley interaction indices $\bm{\Phi}(\bm{x}, \tilde{\bm{\mu}}(\bm{x},S), \Sigma)$ for the \texttt{Chevrolet Volt} and \texttt{Pagani Huayra}, which are computed with respect to the reference point provided by Algorithm~\ref{algorithm:MOE}. The corresponding Shapley values are displayed above the heatmaps.}
    \label{fig:TopGear_interaction}
\end{figure}

\newpage
\subsection{Weather in Vienna}

As a second real-world example, we analyze monthly weather data from the weather station ``Hohe Warte'' in Vienna \citep{weather}. Therefore we we consider 16 numerical attributes, which are described in Table~\ref{tab:weather_parameters} in \ref{section:weather_vienna_parameters}, over a time period spanning from 1955 to 2022.
Furthermore, we restrict our investigation to the three summer months June, July, and August, and compute average values for the considered variables, which yields 68 annual observations for each variable. As for the previous example, we center and scale the data using median and MAD and estimate the covariance using the MCD estimator, before applying the SCD and MOE algorithms using the same setup as before. 

\begin{figure}[ht!]
    \centering
    \includegraphics[width=1\linewidth]{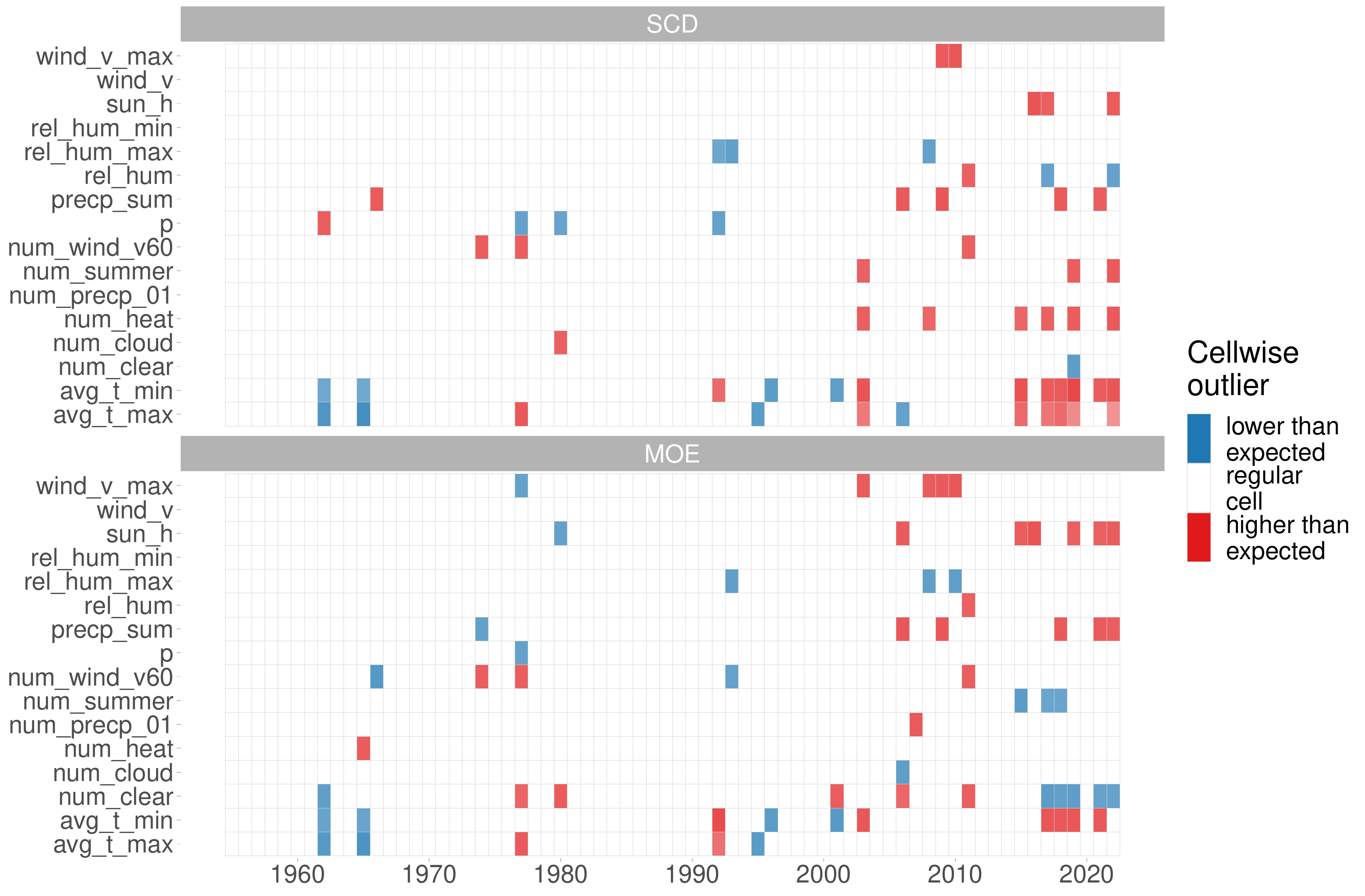}
    \caption{Comparison of outlying cells according to Algorithm~\ref{algorithm:SCD} (top) and Algorithm~\ref{algorithm:MOE} (bottom) for the weather data of Vienna. It is visible in the results of both procedures that the number of anomalies is increasing over the years.}
    \label{fig:weather_tiles}
\end{figure}

Figure~\ref{fig:weather_tiles} displays the outlying cells of the entire 68 years of measurements: The top panel shows the results from the SCD algorithm, and the bottom panel displays those from the MOE algorithm.
Both panels reveal that the number of detected anomalies has increased over the years. 
The SCD procedure further yields results that we would expect to find given that we are currently experiencing an anthropogenic climate change, such as an increasing number of hot days over the years or an increased minimum and maximum mean daily temperature. We emphasize that the SCD procedure results in a global outlyingness measure with respect to the overall mean. On the other hand, the MOE algorithm acts as a local measure: With given values of
the regular cells in a particular year, the outlyingness in the remaining variables is determined.

A more detailed analysis of the results can be made by comparing the Shapley values and pairwise outlyingness scores we obtain from each procedure. Such an analysis is representatively carried out for the year 2021, and the results are displayed in Figure~\ref{fig:weather_interaction}, where we can observe a clear distinction between the results of the SCD and MOE procedure, respectively. 
Both algorithms detect anomalies in the average temperature minimum (\texttt{avg\_t\_min}) and total precipitation (\texttt{precp\_sum}). However, using the local reference point enables the MOE procedure to detect outliers in the number of sun hours (\texttt{sun\_h}) and the number of clear days (\texttt{num\_clear}). 
According to the results of the local MOE procedure given in Figures~\ref{fig:weather_tiles} and \ref{fig:weather_interaction}, the weather of Vienna in 2021 was unusually hot, with more rain than we would expect. When considering the trend of increasing temperature over the years at this specific weather station, we would generally expect fewer sun hours and more clear days than observed in 2021. 

\begin{figure}[ht!]
    \centering
    \includegraphics[width=1\linewidth]{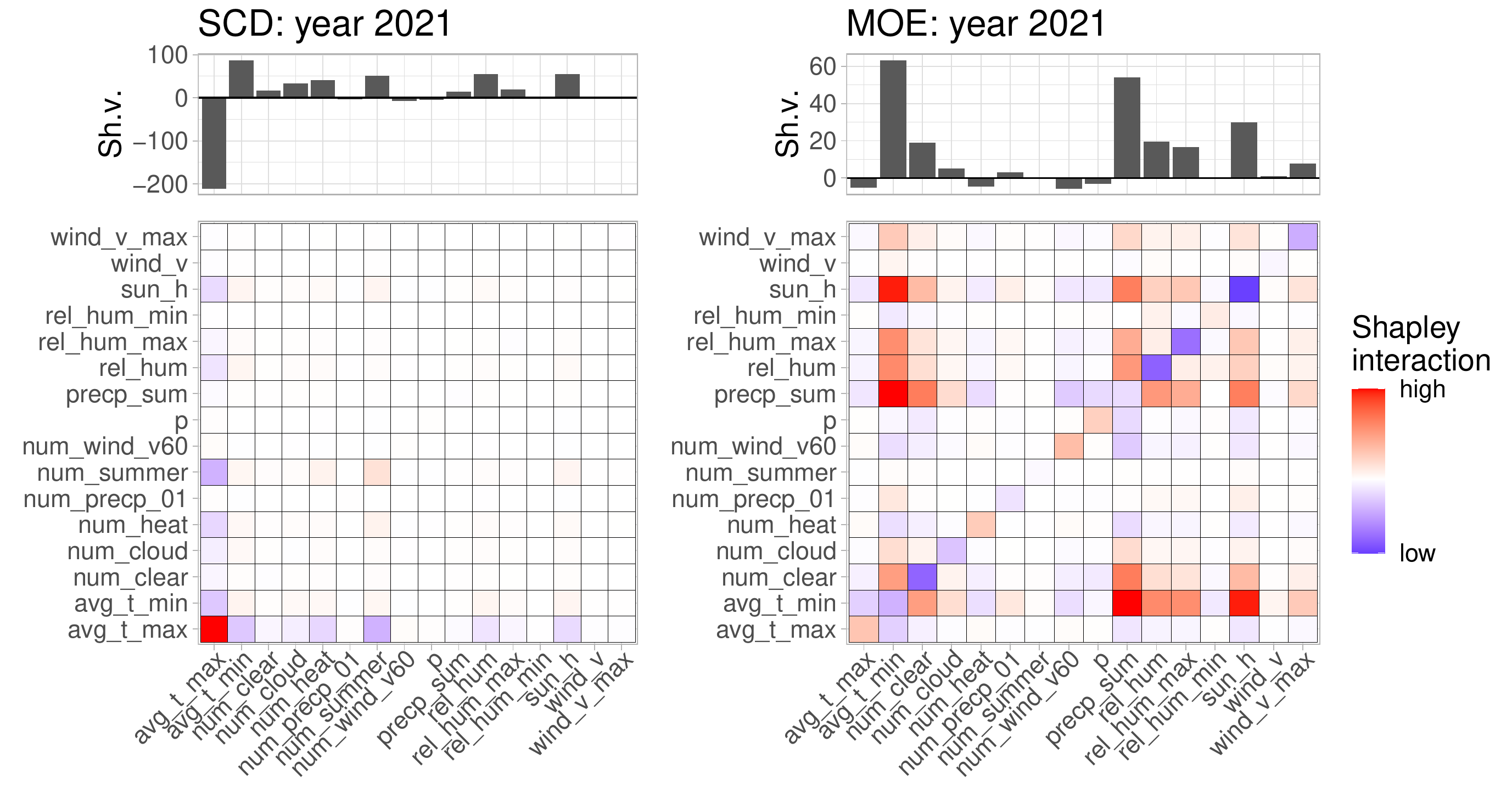}
    \caption{The two graphs in the lower panel show the Shapley interaction indices of the year 2021 for the SCD procedure (left) and the MOE procedure (right). The corresponding Shapley values are displayed above the heatmaps.}
    \label{fig:weather_interaction}
\end{figure}

\section{Discussion and conclusions}
\label{section:conclusion}

This paper introduced Shapley values in connection with Mahalanobis distances for multivariate outlier explanation. 
The Mahalanobis distance is commonly employed for multivariate outlier detection in statistics. Then again, the Shapley value is a concept that originated in cooperative game theory and recently gained popularity in the field of Explainable AI. There it is used to explain the predictions of complex machine learning models by providing information about the contributions of the individual features to a model's prediction.
Combining the Shapley value with the squared Mahalanobis distance enables us to derive outlyingness scores for each coordinate of an observation. Those scores consider all $2^p$ possible combinations of $p$ variables of a single instance and allow us to additively decompose the squared Mahalanobis distance into contributions originating from the individual variables. Without further simplification, the computation would entail evaluating the squared Mahalanobis distance for those $2^p$ combinations, which would pose a substantial computational challenge. However, we showed that our approach leads to a much simpler and computationally efficient form of the Shapley value. 
Moreover, the Shapley interaction indices generalize Shapley values and can be used to derive outlyingness scores for pairs of variables. 

Outlier explanation, and thus identifying the contributions of a variable to the outlyingness of a particular observation, is closely related to cellwise outlyingness, where one aims to identify unusual cells instead of entire observations.
We have adopted cellwise outlyingness into the framework of Shapley values and have proposed two procedures for simultaneous outlier detection and explanation. First, we introduced the SCD procedure as a straightforward implementation of Shapley values for cellwise outlier detection. This algorithm is iteratively replacing anomalous cells with a value towards their mean until the observation is no longer outlying. The more sophisticated MOE procedure takes the information of the non-outlying cells into account and, based on this added input, determines a local reference point.
As a result, one again obtains an additive decomposition of the squared Mahalanobis distance, but with contributions that explain the \textit{local} outlyingness of an observation.

The performance of the two cellwise outlier detection and explanation procedures has been evaluated in simulations and on real-world datasets. It has further been compared to the recently published cellHandler procedure. However, we want to emphasize that the goal of our work is clearly defined as outlier explanation rather than cellwise outlier detection. In particular, Mahalanobis distances rely on a robustly estimated covariance matrix, which has not been in focus in this paper.

We believe that Shapley values are a powerful tool for providing humanly interpretable explanations that allow us to gain further insights into the results of models and methods used in statistics and computer science. They show great potential for further use in this area, especially when a simplification of the computation is possible, as is the case when combining them with Mahalanobis distances. Possible extensions of Shapley values for outlier detection in functional data analysis will be the subject of our future research.

\textbf{Software and data availability:} The methods introduced in this work will be made available in the \texttt{R} package \texttt{ShapleyOutlier} on CRAN, including the weather dataset and a vignette to reproduce the examples presented in Section~\ref{section:applications}.

\section*{Acknowledgements}

The part of the work has been performed in the project AI4CSM under grant agreement No 101007326. The project is co-funded by grants from Germany, Austria, Norway, Belgium, Italy, Netherlands, Czech Republic, Latvia, India and - Electronic Component Systems for European Leadership Joint Undertaking (ECSEL JU).


\appendix

\section{Proof of Theorem \ref{theorem:explain_mahalanobis}}
\label{section:proof_explain_mahalanobis}

\begin{lemma}\label{lemma:shapley_difference}
    The contributions $\Delta_k \md^2(\xhat{S}) = \md^2(\xhat{S\cup\{k\}})-\md^2(\xhat{S})$ 
    can be expressed as
    \begin{equation}
        \Delta_k \md^2(\xhat{S}) = 2(x_k-\mu_k) \left(\sum_{j \in S\cup\{k\}} (x_j-\mu_j) \omega_{jk}\right) - (x_k-\mu_k)^2\omega_{kk}, \label{eq:shapley_difference}
    \end{equation}
    for any subset $S \subseteq P\setminus\{k\}$.
\end{lemma}
\begin{proof}
    \begin{align*}
        \Delta_k \md^2(\xhat{S}) &= \md^2(\xhat{S\cup\{k\}})-\md^2(\xhat{S}) \\
        &= (\xhat{S\cup\{k\}}-\bm{\mu})' \bm{\Sigma}^{-1} (\xhat{S\cup\{k\}}-\bm{\mu}) - (\xhat{S}-\bm{\mu})' \bm{\Sigma}^{-1} (\xhat{S}-\bm{\mu})\\
        &= \sum_{j = 1}^p \sum_{l = 1}^p (\xhatj{S\cup\{k\}}{j}-\mu_j) (\xhatj{S\cup\{k\}}{l}-\mu_l) \omega_{jl} - \sum_{j = 1}^p \sum_{l = 1}^p (\xhatj{S}{j}-\mu_j) (\xhatj{S}{l}-\mu_l) \omega_{jl}\\
        &= \sum_{j \in S\cup\{k\}} \sum_{l \in S\cup\{k\}} (x_j-\mu_j) (x_l-\mu_l) \omega_{jl} - \sum_{j \in S} \sum_{l \in S} (x_j-\mu_j) (x_l-\mu_l) \omega_{jl} \\
        &= \sum_{j \in S\cup\{k\}} (x_k-\mu_k) (x_j-\mu_j) \underbrace{\omega_{kj}}_{=\omega_{jk}} + \sum_{j \in S}  (x_k-\mu_k) (x_j-\mu_j) \omega_{jk}\\
        &= (x_k-\mu_k)^2\omega_{kk} + 2(x_k-\mu_k) \sum_{j \in S} (x_j-\mu_j) \omega_{jk} = \eqref{eq:shapley_difference} 
    \end{align*}
\end{proof}

Now that we have derived a simpler form for the contributions $\Delta_k \md^2(\xhat{S})$, we can use this result to rewrite Equation~\eqref{eq:shapley_md_long} for the $k$-th component of the Shapley value $\phi_k(\bm{x})$. We apply Lemma~\ref{lemma:shapley_difference} in the first step of the proof below, and for the purpose of a simpler notation we write
\begin{align*}
    w(\abs{S}) := \frac{\abs{S}!(p-\abs{S}-1)!}{p!},
\end{align*}
for which $\sum_{S \subseteq P\setminus\{k\}} w(\abs{S}) = 1$ holds.
\begin{proof}[Proof of Theorem \ref{theorem:explain_mahalanobis}]
\begin{align*}
    \phi_k(\bm{x}) &= \sum_{S \subseteq P\setminus\{k\}}
    w(\abs{S})
    \Delta_k \md^2(\xhat{S})\\
    &= \sum_{S \subseteq P\setminus\{k\}} w(\abs{S})\Big((x_k-\mu_k)^2\omega_{kk} + 2(x_k-\mu_k) \sum_{j \in S} (x_j-\mu_j) \omega_{jk}\Big)\\
    &= (x_k-\mu_k)^2\omega_{kk} \Bigg(\underbrace{\sum_{S \subseteq P\setminus\{k\}} w(\abs{S})}_{=1}\Bigg) + 2(x_k-\mu_k) \sum_{S \subseteq P\setminus\{k\}} \left(w(\abs{S}) \sum_{j \in S} (x_j-\mu_j) \omega_{jk}\right) \\
    &=(x_k-\mu_k)^2\omega_{kk} + 2(x_k-\mu_k) \sum_{s=1}^{p-1} \Bigg(w(s) \sum_{\substack{S \subseteq P\setminus\{k\}\\ \abs{S} = s}} \sum_{j \in S} (x_j-\mu_j) \omega_{jk}\Bigg) \\
    & = (x_k-\mu_k)^2\omega_{kk} + 2(x_k-\mu_k) \sum_{s=1}^{p-1} \Bigg(w(s) \binom{p-2}{s-1} \sum_{j \in P\setminus\{k\}} (x_j-\mu_j) \omega_{jk}\Bigg) \\
    & = (x_k-\mu_k)^2\omega_{kk} + 2(x_k-\mu_k) \sum_{s=1}^{p-1} \Bigg(\frac{s }{p(p-1)} \sum_{j \in P\setminus\{k\}} (x_j-\mu_j) \omega_{jk}\Bigg) \\
    & = (x_k-\mu_k)^2\omega_{kk} + (x_k-\mu_k)\sum_{j \in P\setminus\{k\}} (x_j-\mu_j) \omega_{jk} \\
    &= (x_k-\mu_k) \sum_{j \in P} (x_j-\mu_j) \omega_{jk} = (x_k-\mu_k)\left( \sum_{j =1}^p (x_j-\mu_j) \omega_{jk}\right)
\end{align*}

\end{proof}

\section{Proof of Theorem \ref{theorem:explain_mahalanobis_interaction}}
\label{section:proof_explain_mahalanobis_interaction}

\begin{proof}[Proof of Theorem \ref{theorem:explain_mahalanobis_interaction}]
To derive the off-diagonal elements defined in Equation~\eqref{eq:shapley_interaction_long}, we start with rewriting $\Delta_{\{j,k\}} \md^2(\xhat{T}), T \subseteq P \setminus \{j,k\}$, by applying Lemma \ref{lemma:shapley_difference}:
\begin{align*}
    \Delta_{\{j,k\}} \md^2(\xhat{T}) 
    =& \big[\md^2(\xhat{T \cup \{j,k\}}) - \md^2(\xhat{T\cup \{j\}})] - \big[\md^2(\xhat{T\cup \{k\}}) - \md^2(\xhat{T})\big]\\
    =& 2(x_k - \mu_k)\left(\sum_{l \in T\cup\{j\}}(x_l-\mu_l)\omega_{jk} - \sum_{l \in T\cup\{k\}}(x_l-\mu_l)\omega_{jk}\right) +  2(x_k-\mu_k)^2\omega_{kk}\\
    =& 2(x_k - \mu_k)\left((x_j-\mu_j) \omega_{jk} - (x_k-\mu_k) \omega_{kk}\right) +  2(x_k-\mu_k)^2\omega_{kk}\\
    =& 2(x_k - \mu_k)(x_j-\mu_j) \omega_{jk}
\end{align*}
Moving on, we plug the result into the formula for $\Phi_{jk}$ for $j \neq k$, given in Equation~\eqref{eq:shapley_interaction_long}, and we obtain
\begin{align*}
    \Phi_{jk} &= \sum_{T \subseteq P \setminus \{j,k\}} \frac{t!(p-t-2)!}{(p-1)!} \Delta_{\{j,k\}} \md^2(\xhat{T})\\
    &= \sum_{T \subseteq P \setminus \{j,k\}} \frac{t!(p-t-2)!}{(p-1)!} 2(x_k - \mu_k)(x_j-\mu_j) \omega_{jk}\\
    &= 2(x_k - \mu_k)(x_j-\mu_j)\omega_{jk},
\end{align*}
where the last equality is obtained by following the same structure as in the proof of Theorem \ref{theorem:explain_mahalanobis}. Finally, we have to derive the diagonal elements $\Phi_{jj}$ given by 
\begin{align*}
    \Phi_{jj} &= \phi_j - \sum_{k \neq j} \Phi_{jk} \\
    &= (x_j-\mu_j) \sum_{k =1}^p (x_k-\mu_k) \omega_{jk} - 2(x_j-\mu_j)\sum_{k \neq j} (x_k - \mu_k) \omega_{jk}\\
    &=  (x_j-\mu_j)^2 \omega_{jj} - (x_j-\mu_j)\sum_{k \neq j} (x_k - \mu_k) \omega_{jk}.
\end{align*}
\end{proof}

\subsection{Higher order interactions}
\label{subsection:proof_explain_mahalanobis_interaction}
\begin{proof}
To show that all interactions of order three or higher are zero, it is sufficient to show that for the three-way interactions the \textit{set function derivative} $\Delta_{\{j,k,l\}} \md^2(\xhat{T})$ is zero for all $T \subseteq P \setminus \{j,k,l\}$. This follows from the iterative definition of the set function derivative for $S \cap \{j\} = \emptyset$ \citep{Grabisch2016}, which is given by
\begin{align*}
    \Delta_{S \cup \{j\}} \md^2(\xhat{T}) = \Delta_{S}( \Delta_{j} \md^2(\xhat{T})).
\end{align*}
Hence, to show that all Shapley interaction indices 
\begin{align*}
    I_{Sh}(v,S)
    = \sum_{T \subseteq P \setminus S} \frac{t!(p-t-s)!}{(p-s+1)!} \Delta_S v(T),
\end{align*}
with $\abs{S} \geq 3$ are zero, we only have to prove that $\Delta_{\{j,k,l\}} \md^2(\xhat{T}) = 0, \forall T \subseteq P \setminus \{j,k,l\}$. For this purpose, we first rewrite the above expression and then apply Lemma~\ref{lemma:shapley_difference}:
\begin{align*}
    \Delta_{\{j,k,l\}} \md^2(\xhat{T}) =&
    \begin{aligned}[t]
        &- \md^2(\xhat{T\cup \{j,k,l\}}) \\
        &+ \md^2(\xhat{T\cup \{j,k\}}) + \md^2(\xhat{T\cup \{j,l\}}) + \md^2(\xhat{T\cup \{k,l\}}) \\
        &- \md^2(\xhat{T\cup \{j\}}) - \md^2(\xhat{T\cup \{k\}}) - \md^2(\xhat{T\cup \{l\}}) \\
        &+ \md^2(\xhat{T})
    \end{aligned}\\
    =&
    \begin{aligned}[t]
        & - \big[\md^2(\xhat{T\cup \{j,k,l\}}) - \md^2(\xhat{T\cup \{k,l\}})\big] \\
        & + \big[\md^2(\xhat{T\cup \{j,l\}}) - \md^2(\xhat{T\cup \{l\}})\big]\\
        & + \big[\md^2(\xhat{T\cup \{j,k\}}) - \md^2(\xhat{T\cup \{j\}}) - \md^2(\xhat{T\cup \{k\}}) + \md^2(\xhat{T}) \big]
    \end{aligned}\\
    =&
    \begin{aligned}[t]
        & - \big[(x_j-\mu_j)^2\omega_{jj} + 2(x_j - \mu_j)\sum_{m \in T\cup\{k,l\}}(x_m-\mu_m)\omega_{jm}\big]\\
        & + \big[(x_j-\mu_j)^2\omega_{jj} + 2(x_j - \mu_j)\sum_{m \in T\cup\{l\}}(x_j-\mu_j)\omega_{jk}\big] \\
        & + \big[2(x_k - \mu_k)(x_j-\mu_j)\big]
    \end{aligned}\\
    =& -2(x_k - \mu_k)(x_j-\mu_j) + 2(x_k - \mu_k)(x_j-\mu_j) = 0
\end{align*}
\end{proof}

\newpage
\section{Weather in Vienna - Parameters}
\label{section:weather_vienna_parameters}
Table~\ref{tab:weather_parameters} shows the parameter descriptions for the \textit{Weather in Vienna} dataset, which have been adapted and translated from \cite{weather}. 

\begin{table}[!h]
    \caption{Description of the parameters of the \textit{Weather in Vienna} dataset}
    \centering
    \begin{tabular}[t]{ll}
        \toprule
        Parameter & Desciption\\
        \midrule
        \texttt{avg\_t\_max} & \begin{minipage}[t]{0.8\textwidth}Mean daily maximum air temperature in °C\end{minipage} \\
        \texttt{avg\_t\_min} & \begin{minipage}[t]{0.8\textwidth}Mean daily minimum air temperature in °C\end{minipage} \\
        \texttt{num\_summer} & \begin{minipage}[t]{0.8\textwidth}Number of summer days \newline (days with a temperature maximum $t_{\text{max}} \geq 25.0$ °C)\end{minipage} \\
        \texttt{num\_heat} & \begin{minipage}[t]{0.8\textwidth}Number of hot days \newline (days with a temperature maximum $t_{\text{max}} \geq 30.0$ °C)\end{minipage} \\
        \addlinespace
        \texttt{p} & \begin{minipage}[t]{0.8\textwidth}Daily mean air pressure in hPa \newline(mean of all measurements at 7 a.m., 2 p.m., 7 p.m. CET; before 1971 9 p.m. instead of 7 p.m.)\end{minipage} \\
        \texttt{sun\_h} & \begin{minipage}[t]{0.8\textwidth}Monthly total sunshine duration in hours\end{minipage} \\
        \texttt{num\_clear} & \begin{minipage}[t]{0.8\textwidth}Number of clear days (daily mean cloudiness $< 20/100$)\end{minipage} \\
        \texttt{num\_cloud} & \begin{minipage}[t]{0.8\textwidth}Number of cloudy days (daily mean cloudiness $> 80/100$)\end{minipage} \\
        \addlinespace
        \texttt{rel\_hum} & \begin{minipage}[t]{0.8\textwidth}Daily mean relative humidity in percent \newline(2 x RH7 mean + RH14 mean + RH19 mean)/4; before 1971 9 p.m. instead of 7 p.m.)\end{minipage} \\
        \texttt{rel\_hum\_max} & \begin{minipage}[t]{0.8\textwidth}Relative humidity maximum in percent\end{minipage} \\
        \texttt{rel\_hum\_min} & \begin{minipage}[t]{0.8\textwidth}Relative humidity minimum in percent\end{minipage} \\
        \texttt{wind\_v} & \begin{minipage}[t]{0.8\textwidth}Monthly average wind speed in km/h\end{minipage} \\
        \addlinespace
        \texttt{num\_wind\_v60} & \begin{minipage}[t]{0.8\textwidth}Number of days with wind peaks $\geq 60$ km/h\end{minipage} \\
        \texttt{wind\_v\_max} & \begin{minipage}[t]{0.8\textwidth}Maximum wind speed in km/h\end{minipage} \\
        \texttt{precp\_sum} & \begin{minipage}[t]{0.8\textwidth}Monthly total precipitation in mm\end{minipage} \\
        \texttt{num\_precp\_01} & \begin{minipage}[t]{0.8\textwidth}Number of days with precipitation $\geq 0.1$ mm \end{minipage}\\
        \bottomrule
    \end{tabular}
    \label{tab:weather_parameters}
\end{table}

\bibliographystyle{apalike}
\bibliography{literature.bib}

\begin{thebibliography}{}

\bibitem[Agostinelli et~al., 2015]{Agostinelli2015}
Agostinelli, C., Leung, A., Yohai, V., and Zamar, R. (2015).
\newblock Robust estimation of multivariate location and scatter in the
  presence of cellwise and casewise contamination.
\newblock {\em Test}, 24:441--461.

\bibitem[Alfons, 2021]{Alfons2021}
Alfons, A. (2021).
\newblock {robustHD}: An {R} package for robust regression with
  high-dimensional data.
\newblock {\em Journal of Open Source Software}, 6(67):3786.

\bibitem[Alqallaf et~al., 2009]{alqallaf2009propagation}
Alqallaf, F., Van~Aelst, S., Yohai, V.~J., and Zamar, R.~H. (2009).
\newblock Propagation of outliers in multivariate data.
\newblock {\em The Annals of Statistics}, pages 311--331.

\bibitem[Biecek and Burzykowski, 2021]{brzemyslaw_2021}
Biecek, P. and Burzykowski, T. (2021).
\newblock {\em {Explanatory Model Analysis}}.
\newblock Chapman and Hall/CRC, New York.

\bibitem[Chandola et~al., 2009]{Chandola2009}
Chandola, V., Banerjee, A., and Kumar, V. (2009).
\newblock Anomaly detection: A survey.
\newblock {\em ACM Comput. Surv.}, 41(3).

\bibitem[Debruyne et~al., 2019]{debruyne2019}
Debruyne, M., Höppner, S., Serneels, S., and Verdonck, T. (2019).
\newblock Outlyingness: Which variables contribute most?
\newblock {\em Statistics and Computing}, 29:707–723.

\bibitem[Filzmoser et~al., 2014]{Filzmoser2014}
Filzmoser, P., Ruiz-Gazen, A., and Thomas-Agnan, C. (2014).
\newblock Identification of local multivariate outliers.
\newblock {\em Statistical Papers}, 55.

\bibitem[Fujimoto et~al., 2006]{Fujimoto2006}
Fujimoto, K., Kojadinovic, I., and Marichal, J.-L. (2006).
\newblock Axiomatic characterizations of probabilistic and
  cardinal-probabilistic interaction indices.
\newblock {\em Games and Economic Behavior}, 55:72--99.

\bibitem[Grabisch, 2016]{Grabisch2016}
Grabisch, M. (2016).
\newblock {\em Set Functions, Games and Capacities in Decision Making}.
\newblock Springer Publishing Company, Incorporated, 1st edition.

\bibitem[Grabisch and Roubens, 1999]{Grabisch1999}
Grabisch, M. and Roubens, M. (1999).
\newblock An axiomatic approach to the concept of interaction among players in
  cooperative games.
\newblock {\em International Journal of Game Theory}, 28:547--565.

\bibitem[Grubbs, 1969]{Grubbs1969}
Grubbs, F.~E. (1969).
\newblock Procedures for detecting outlying observations in samples.
\newblock {\em Technometrics}, 11(1):1--21.

\bibitem[Lundberg et~al., 2020]{lundberg2020}
Lundberg, S.~M., Erion, G., Chen, H., Degrave, A., Prutkin, J.~M., Nair, B.,
  Katz, R., Himmelfarb, J., Bansal, N., Lee, S.-I., and et~al. (2020).
\newblock From local explanations to global understanding with explainable ai
  for trees.
\newblock {\em Nature Machine Intelligence}, 2(1):56–67.

\bibitem[Lundberg et~al., 2018]{lundberg2019}
Lundberg, S.~M., Erion, G.~G., and Lee, S.-I. (2018).
\newblock Consistent individualized feature attribution for tree ensembles.
\newblock {\em arXiv preprint arXiv:1802.03888}.

\bibitem[Lundberg and Lee, 2017]{lundberg2017}
Lundberg, S.~M. and Lee, S.-I. (2017).
\newblock A unified approach to interpreting model predictions.
\newblock In Guyon, I., Luxburg, U.~V., Bengio, S., Wallach, H., Fergus, R.,
  Vishwanathan, S., and Garnett, R., editors, {\em Advances in Neural
  Information Processing Systems 30}, pages 4765--4774. Curran Associates, Inc.

\bibitem[Mahalanobis, 1936]{mahalanobis1936}
Mahalanobis, P.~C. (1936).
\newblock On the generalized distance in statistics.
\newblock {\em Proceedings of the National Institute of Sciences (Calcutta)},
  2:49--55.

\bibitem[Molnar, 2019]{molnar2019}
Molnar, C. (2019).
\newblock {\em Interpretable Machine Learning}.
\newblock \url{https://christophm.github.io/interpretable-ml-book/}.

\bibitem[Peters, 2008]{peters2008}
Peters, H. (2008).
\newblock {\em Game Theory}.
\newblock Springer, Berlin Heidelberg.

\bibitem[Raymaekers and Rousseeuw, 2021]{Raymaekers2021}
Raymaekers, J. and Rousseeuw, P. (2021).
\newblock Handling cellwise outliers by sparse regression and robust
  covariance.
\newblock {\em Journal of Data Science, Statistics, and Visualisation}, 1.

\bibitem[Raymaekers and Rousseeuw, 2022]{Raymaekers2022}
Raymaekers, J. and Rousseeuw, P.~J. (2022).
\newblock The cellwise minimum covariance determinant estimator.

\bibitem[Ribeiro et~al., 2016]{Ribeiro2016}
Ribeiro, M., Singh, S., and Guestrin, C. (2016).
\newblock “why should i trust you?”: Explaining the predictions of any
  classifier.
\newblock pages 97--101.

\bibitem[Rousseeuw, 1985]{Rousseeuw1985}
Rousseeuw, P. (1985).
\newblock Multivariate estimation with high breakdown point.
\newblock {\em Mathematical Statistics and Applications Vol. B}, pages
  283--297.

\bibitem[Rousseeuw and Zomeren, 1990]{Rousseeuw1990}
Rousseeuw, P. and Zomeren, B. (1990).
\newblock Unmasking multivariate outliers and leverage points.
\newblock {\em Journal of The American Statistical Association - J AMER STATIST
  ASSN}, 85:633--639.

\bibitem[Rousseeuw and Bossche, 2018]{Rousseeuw2018}
Rousseeuw, P.~J. and Bossche, W. V.~D. (2018).
\newblock Detecting deviating data cells.
\newblock {\em Technometrics}, 60(2):135--145.

\bibitem[Shapley, 1953]{shapley1953}
Shapley, L.~S. (1953).
\newblock A value for n-person games.
\newblock {\em Contributions to the Theory of Games}, 2(28):307--317.

\bibitem[{Stadt Wien}, 2022]{weather}
{Stadt Wien} (2022).
\newblock {Monthly data from the weather station Hohe Warte since April 1872 -
  Vienna}.

\bibitem[Sundararajan et~al., 2020]{sundararajan2020shapley}
Sundararajan, M., Dhamdhere, K., and Agarwal, A. (2020).
\newblock The shapley taylor interaction index.
\newblock In {\em International Conference on Machine Learning}, pages
  9259--9268. PMLR.

\bibitem[Young, 1985]{young1985}
Young, H. (1985).
\newblock Monotonic solutions of cooperative games.
\newblock {\em International Journal of Game Theory}, 14:65--72.

\bibitem[Zimek and Filzmoser, 2018]{Zimek2018}
Zimek, A. and Filzmoser, P. (2018).
\newblock There and back again: Outlier detection between statistical reasoning
  and data mining algorithms.
\newblock {\em Wiley Interdisciplinary Reviews: Data Mining and Knowledge
  Discovery}, 8:e1280.

\bibitem[Štrumbelj and Kononenko, 2010]{strumbelj2010}
Štrumbelj, E. and Kononenko, I. (2010).
\newblock An efficient explanation of individual classifications using game
  theory.
\newblock {\em Journal of Machine Learning Research}, 11:1--18.

\bibitem[Štrumbelj and Kononenko, 2014]{strumbelj2014}
Štrumbelj, E. and Kononenko, I. (2014).
\newblock Explaining prediction models and individual predictions with feature
  contributions.
\newblock {\em Knowledge and Information Systems}, 41:647–665.

\end{thebibliography}

\end{document}